\documentclass[acmtog, nonacm]{acmart}

\usepackage{booktabs} %

\usepackage{geometrycollective}

\renewcommand{\T}[1]{T^{#1}}
\newcommand{\V}[1]{V^{#1}}
\newcommand{\E}[1]{E^{#1}}
\newcommand{\F}[1]{F^{#1}}
\renewcommand{\H}[1]{H^{#1}}
\newcommand{\Tl}[1]{\ell^{#1}}

\newcommand{\Vinserted}{{\V{}}^\star}

\newcommand{\Tpos}[2]{q^{#1}_{#2}}

\newcommand{\halfedge}[1]{\smash{\overset{\raisebox{3pt}{$\rightharpoonup$}}{\smash{#1}}}}
\newcommand{\he}[1]{\smash{\overset{\raisebox{3pt}{\tiny$\,\rightharpoonup$}}{\smash{#1}}}}

\algnewcommand{\LeftComment}[1]{\textcolor{commentblue}{\(\triangleright\)\textit{#1}}}

\newcommand{\corner}[2]{\smash{\text{\raisebox{-2pt}{$\overset{#2}{\rule{0pt}{2pt}\smash{\!_{#1}}}$}}}}

\newcommand{\angletriplet}[3]{\angle_{{#1}{#2}{#3}}}

\newcounter{algo}
\newenvironment{algo}[1]
{ \refstepcounter{algo}\noindent\rule{\columnwidth}{1.25pt}\vspace{-.2\baselineskip} \\ \textbf{Algorithm~\thealgo} #1\vspace{-.55\baselineskip} \\ \noindent\rule{\columnwidth}{.5pt}\vspace{-1.2\baselineskip} }
{ \vspace{-.8\baselineskip}\noindent\rule{\columnwidth}{.5pt}\vspace{-\baselineskip} }

\acmJournal{TOG}
\acmVolume{99}
\acmNumber{99}
\acmArticle{99}

\setcopyright{acmcopyright}

\copyrightyear{2021}
\acmYear{2021}
\acmMonth{99}
\acmDOI{}

\begin{document}
\title{Integer Coordinates for Intrinsic Geometry Processing}
\author{Mark Gillespie}
\author{Nicholas Sharp}
\author{Keenan Crane}
\affiliation{%
  \institution{Carnegie Mellon University}
  \streetaddress{5000 Forbes Ave}
  \city{Pittsburgh}
  \state{PA}
  \postcode{15213}}

\renewcommand\shortauthors{Gillespie, Sharp, and Crane}

\begin{abstract}

  \vspace*{1em}

  In this work, we present a general, efficient, and provably robust representation for intrinsic triangulations.
  These triangulations have emerged as a powerful tool for robust geometry processing of surface meshes, taking a low-quality mesh and retriangulating it with high-quality intrinsic triangles.
  However, existing representations either support only edge flips, or do not offer a robust procedure to recover the common subdivision, that is, how the intrinsic triangulation sits along the original surface.
   To build a general-purpose robust structure, we extend the framework of \emph{normal coordinates}, which have been deeply studied in topology, as well as the more recent idea of \emph{roundabouts} from geometry processing, to support a variety of mesh processing operations like vertex insertions, edge splits, \etc{}
   The basic idea is to store an integer per mesh edge counting the number of times a curve crosses that edge.
  We show that this paradigm offers a highly effective representation for intrinsic triangulations with strong robustness guarantees.
  The resulting data structure is general and efficient, while offering a guarantee of always encoding a valid subdivision.
  Among other things, this allows us to generate a high-quality intrinsic Delaunay refinement of all manifold meshes in the challenging Thingi10k dataset for the first time.
  This enables a broad class of existing surface geometry algorithms to be applied out-of-the-box to low-quality triangulations.

  \vspace*{1.5em}
\end{abstract}



\maketitle

\begin{figure}
  \vspace*{3em}
  \centering
  \includegraphics{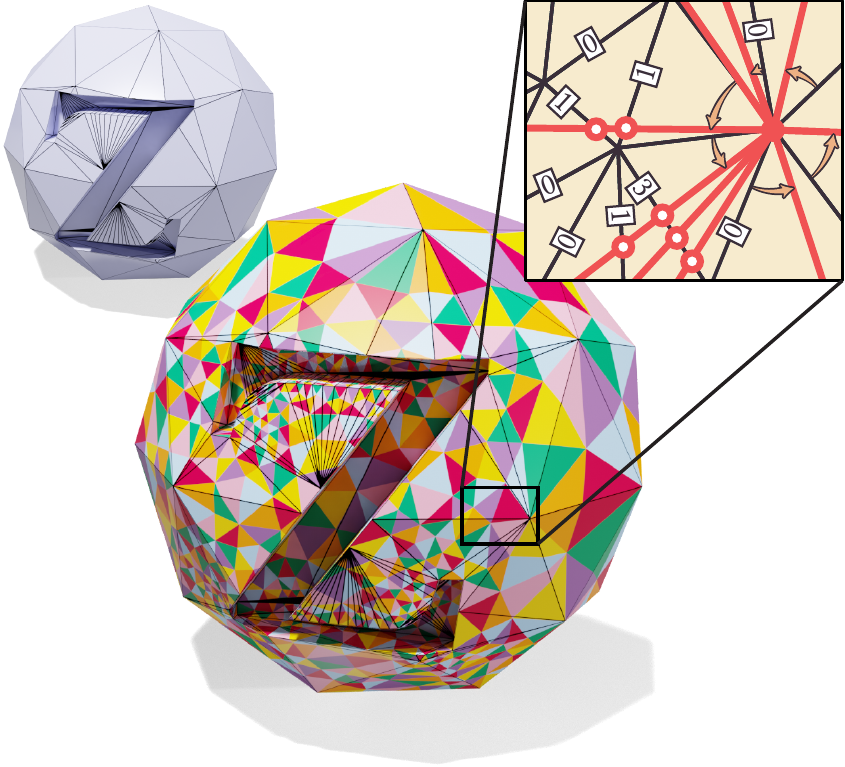}
   \vspace*{0.5em}
   \caption{
   We extend the machinery of normal coordinates beyond just edge flips, to enable a broader set of local mesh operations such as vertex insertion.  By doing so, we dramatically improve the robustness of algorithms like Delaunay refinement in the intrinsic setting.\label{fig:teaser}}
\end{figure}

\vspace*{1.5em}

\section{Introduction and Related Work}
\label{sec:Introduction}

Geometric data plays a growing role in applications from computational fabrication to to autonomous driving to augmented reality, but data in these applications is increasingly difficult to deal with due to the poor quality of meshes generated by non-expert users, or by algorithms targeted at visualization rather than mesh processing---there have hence been significant recent efforts to make geometric algorithms more robust~\cite{Zhou:2016:MAS,Hu:2018:TMW,Sellan:2019:SGP,Sawhney:2020:MCG}.  One basic tool is to remesh the input to obtain higher-quality elements, but for this approach to work on difficult, near-degerate inputs, remeshing algorithms must themselves be extremely robust.  Moreover, traditional approaches to remeshing based on vertex positions in \(\mathbb{R}^n\) must negotiate the trade-off between mesh size, the quality of mesh elements, and geometric approximation of the input domain.

\paragraph{Intrinsic Triangulations.} A promising idea is to approach geometry processing from the \emph{intrinsic} point of view: rather than considering the embedding of the geometry in space, one focuses only on point-to-point distances along the surface, as encoded by the edge lengths of a triangulation.  This perspective is quite natural for problems in geometry processing and scientific computing, since many objects in these domains are themselves intrinsic---for instance, the Laplace-Beltrami operator, which appears in numerous algorithms and fundamental partial differential equations (PDEs).  In this paper we consider so-called \emph{intrinsic triangulations}, whose edges no longer need to be straight line segements in Euclidean space, but can instead be any straight or \emph{geodesic} path across the input polyhedron (see \figref{IntrinsicTriangulationExample}).  This construction still captures the input geometry exactly, but provides a dramatically larger space of possible triangulations, lending enormous flexibility to geometric algorithms.  For instance, it de-couples the quality of elements used for simulation from the elements used to describe the geometry, side-stepping the trade-off encountered in traditional meshing.  Moreover, this perpective enables the input polyhedral surface to serve as a background domain---analogous to the Euclidean plane in traditional computational geometry---allowing one to ``port'' trusted algorithms from the plane to curved surfaces.  Most importantly, by encapsulating all this machinery in an interface that resembles an ordinary mesh, one can provide robustness as a subroutine: rather than make existing algorithms more robust one by one, we can transform the input into an intrinsic triangulation, execute an ordinary (``non-robust'') algorithm, and then read off the reults in a variety of ways.

\begin{figure}
  \includegraphics{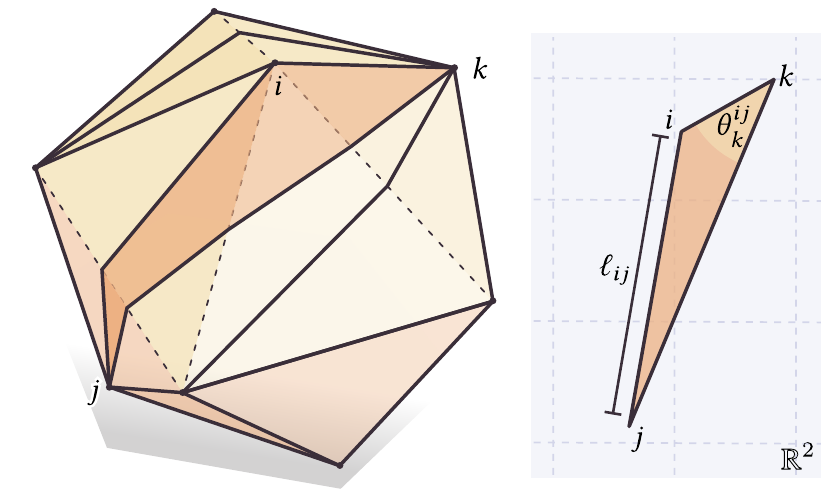}
  \caption{Edges of an intrinsic triangulation are allowed to be \emph{geodesic paths} along a surface (\figloc{left}). The faces of such a triangulation can be laid out in the plane as ordinary triangles (\figloc{right}). \label{fig:IntrinsicTriangulationExample}}
\end{figure}

A remaining impediment to making the intrinsic approach truly reliable is to develop \emph{data structures} for intrinsic triangulations that provide all the expected operations from standard mesh processing, while simultaneously providing strong guarantees of correctness.  The basic challenge is encoding the correpsondence between the input mesh \(\T0\) and an intrinsic triangulation \(\T1\) sitting atop it, so that data on one triangulation can be transferred to the other.  The first such data structure was the \emph{overlay mesh} of \citet{Fisher:2006:IDT}. The overlay uses a halfedge mesh decorated with special vertex and edge attributes to maintain the \emph{common subdivision} of \(\T0\) and \(\T1\), \ie, the polygon mesh obtained by ``slicing up'' the underlying surface along the edges of both \(\T0\) and \(\T1\).  This approach guarantees correct connectivity, but ordinarily-local operations such as edge flips become non-local and expensive to evaluate---moreover, edge flips are the only operation supported by this data structure.  \citet{Sharp:2019:NIT} instead encode the correspondence implicitly by storing so-called \emph{signposts} at vertices, which give the direction and length of each intrinsic edge.  This approach is somewhat complementary to the overlay mesh: local mesh operations are now cheap to evaluate, but the encoding of connectivity now depends on floating-point values, and is hence not guaranteed to be correct.  For instance, when tracing out intrinsic edges small floating point errors can cause one to ``miss'' the target vertex.  A key development here, however, was extending intrinsic triangulations to operations beyond edge flips.

\begin{figure}
  \includegraphics{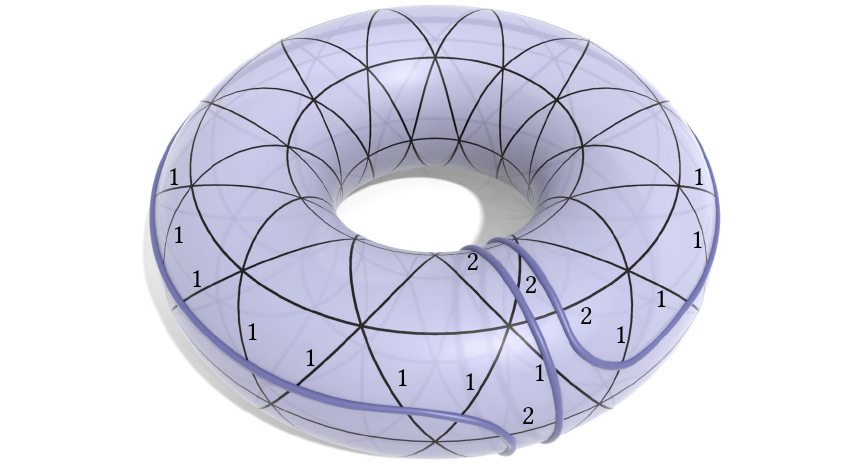}
  \caption{We build on the idea of \emph{normal coordinates} count how many times a curve crosses each edge of a triangulation. \label{fig:NormalCoordinatesTorus}}
\end{figure}

\paragraph{Integer-Based Encoding.} In this paper, we introduce an integer-based data structure for intrinsic triangulations that offers the best of both worlds: an implicit encoding of correspondence that supports fast local operations, but which is also guaranteed to correctly describe connectivity.  Like the signposts, our data structure also supports a wide variety of local mesh operations (\secref{AlgorithmsAndDataStructures}).  As demonstrated in \secref{Evaluation}, we get dramatically improved robustness for difficult tasks, \eg{}, we achieve a 100\% success rate for extracting a high-quality \emph{Delaunay refinement} of low-quality input data.  In turn, any algorithm that relies on a high-quality triangulation (\eg{}, for solving PDEs) can immediately benefit from this improved robustness.  For instance, in \secref{Applications} we observe improved robustness for computing geodesic distance, local parameterization, constructing geodesics, and finding smooth vector fields.

The basic starting point for our data structure is the concept of \emph{normal coordinates} from geometric topology (not to be confused with \emph{geodesic normal coordinates} from Riemannian geometry).  However, we must augment this construction in several ways in order to make it suitable for geometry processing.  As detailed in \secref{IntegerCoordinates}, the basic idea of normal coordinates is to simply count how many times each edge of a triangulation is crossed by some curve (\figref{NormalCoordinatesTorus}).  Such coordinates were originally developed to study not curves, but rather embeddings of surfaces in 3-manifolds~\cite{Kneser:1929:GFDM,Haken:1961:TNIK,Hass:2020:AIG}, and subsequently appear in several places in mathematics (\eg{}, for studying the \emph{mapping class group}~\cite{Farb:2011:APM}), including significant work on algorithms~\cite{Bell:2015:RMC,Bell:2018:F,Schaefer:2008:CDT}.  In theoretical computer science, normal coordinates are also viewed as a means of ``compressing'' curves, \eg{}, the total number of bits required to store a long winding curve can be exponentially smaller than storing explicit segments along the curve~\cite{Erickson:2013:TCC}.

One challenge with using normal coordinates for geometry processing is that existing literature rarely considers operations beyond edge flips: the little that does considers only closed loops \eg{}~\cite[Section 5.4]{Schaefer:2002:ANC}, whereas curves that terminate at vertices are absolutely essential for encoding triangulations.  A second issue is that normal coordinates alone are not enough to uniquely identify curves implied by the coordinates with logical edges of a mesh.  Very recently, \citet[Section 5.2]{Gillespie:2021:DCE} proposed a solution to this issue using what they call \emph{roundabouts}, but again do not consider operations beyond edge flips.  Third, whereas most literature assumes that normal coordinates encode homotopy classes of curves in a purely topological setting (or perhaps hyperbolic geodesics), we must make a significant departure this perspective and assume that the normal coordinates encode a triangulation of a Euclidean polyhedron by geodesic edges.  This distinction is important since, in general, not all normal coordinates describe a valid Euclidean geodesic triangulation.  Considering this special case in turn enables us to establish procedures not previously seen.

\paragraph{Contributions.} Overall, we make the following contributions:
\begin{itemize}
  \item We describe normal coordinates as a representation for general intrinsic triangulations, including the case where vertices have been added to the triangulation.
  \item We extend integer-based data structures for geometric intrinsic triangulations to include local operations beyond edge flips.
  \item We prove the correctness and quality of a Delaunay refinement algorithm for intrinsic triangulations of surfaces without boundary.
  \item We also extend intrinsic Delaunay refinement to surfaces with boundary.
  \item We introduce a new, more accurate way of transferring functions between bases on different triangulations.
\end{itemize}

We also experimentally validate the robustness of our technique, including generating intrinsic Delaunay refinements for all manifold meshes in the Thingi10k dataset, and demonstrate robustness for a variety of basic algorithms from geometry processing.

\section{Notation and Conventions}
\label{sec:NotationAndConventions}

\subsection{Connectivity}
\label{sec:Connectivity}
\setlength{\columnsep}{.5em}
\setlength{\intextsep}{.5em}
\begin{wrapfigure}[8]{r}{43pt}
  \vspace{-.75\baselineskip}\includegraphics{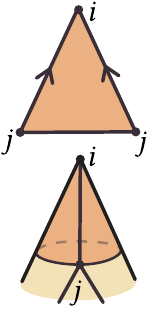}
\end{wrapfigure}
Throughout we assume that our domain is an oriented manifold surface \(M\), possibly with boundary.  We write \(T = (V,E,F)\) to denote a triangulation of \(M\) with vertices \(V\), edges \(E\), and faces \(F\).  In general we allow triangulations that are not be simplicial, but can instead be a \(\Delta\)-complex in the sense of \citet[Section 2.1]{Hatcher:2002:AT}---we allow, \eg{}, two edges of the same triangle to be glued together (see inset).  We will refer to vertices \(i \in V\), edges \(ij \in E\), and faces \(ijk \in F\) by one, two, or three indices, \resp{}.
Note that as our triangulations need not be simplicial, the vertices \(i,j\) of an edge may not be distinct, and do not necessarily identify the edge---there may be multiple edges between \(i\) and \(j\).

\setlength{\columnsep}{.5em}
\setlength{\intextsep}{.5em}
\begin{wrapfigure}[4]{r}{52pt}
  \vspace{-.75\baselineskip}\includegraphics{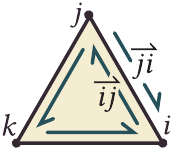}
\end{wrapfigure}
\noindent We will refer to oriented \emph{halfedges} \(\halfedge{ij} \in H\), and will use \(\smash{u_i^{jk}}\) to denote a value \(u\) at corner \(i\) of triangle \(ijk\), and \(u_{\he{ij}}\) to denote a value \(u\) at halfedge \(\halfedge{ij}\).
Throughout, we consider a fixed input triangulation \(\T0\) of \(M\), as well as a dynamic \emph{intrinsic triangulation} \(\T1\) sitting atop \(M\).
\(\T1\) must contain all vertices of \(\T0\), \ie{} \(\V1 \supseteq \V0\), but may include additional inserted vertices which we denote by \(\Vinserted := \V1 \setminus \V0\).
We will generally use indices \(a,b,c\) for vertices of \(\T0\) and indices \(i,j,k\) for vertices of \(\T1\) (which may also be in \(\T0\)).

\subsection{Geometry}
\label{sec:Geometry}
The geometry of a triangulation is determined by a collection of edge lengths \(\ell: E \to \mathbb{R}_{>0}\) satisfying the triangle inequalities in each face.
For instance, we typically begin by assigning \(\E0\) and \(\E1\) the same edge lengths \(\smash{\Tl{}_{ij} = |f_j - f_i|}\) determined by vertex coordinates \(f: \V{} \to \mathbb{R}^3\). While the lengths \(\Tl0\) remain fixed, the edge lengths \(\Tl1\) of \(\T1\) may change due to operations like intrinsic edge flips (\secref{EdgeFlip}).
Even after changing the edge lengths, each individual triangle \(ijk \in \F1\) can always be drawn as an ordinary triangle in the Euclidean plane, allowing us to compute quantities such as face areas or corner angles.
In general, we will not need to simultaneously embed all triangles of $\T1$ in \(\mathbb{R}^3\). Finally, we use \(\exp_x: T_x M \to M\) \hfill
\setlength{\columnsep}{.5em}
\setlength{\intextsep}{.5em}
\begin{wrapfigure}{r}{75pt}
  \vspace{-.25\baselineskip}\includegraphics{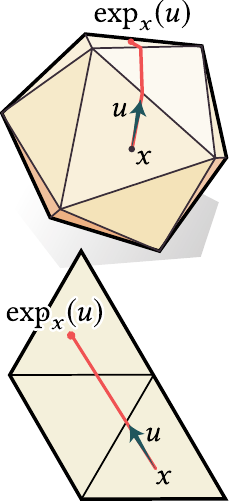}
\end{wrapfigure}
to denote the exponential map at \(x\).  Given a tangent vector \(u\) at \(x\), \(exp_x(u)\) is the point reached by walking straight along the surface in the direction of \(u\) for a distance \(|u|\) (see inset).  (In practice this can be implemented as in \cite[Section 3.2.2]{Sharp:2019:NIT}).

Points \(x \in M\) can be expressed in \emph{barycentric coordinates} relative to some simplex (vertex, edge, or triangle) of a triangulation \(\T{}\), \eg, a point in a triangle \(ijk\) is given by three coordinates \(u_i, u_j, u_k \in [0,1]\) such that \(u_i + u_j + u_k = 1\).
For vertices, we set the single coordinate \(u_i\) to 1.
Importantly, we will need to encode points \(x\) with respect to \emph{two different} triangulations \(\T0\) and \(\T1\), using barycentric coordinates \(u\) and \(v\), \resp{} (\figref{BarycentricCoordinates}).  Note that for vertices that are shared by both triangulations, we have \(u_i = v_i = 1\).
We will use \(\Tpos0{}\) to denote a location on \(\T0\) represented by simplex along with a barycentric coordinate, and similarly will use \(\Tpos1{}\) for a location on \(\T1\).

\begin{figure}
  \includegraphics[width=\columnwidth]{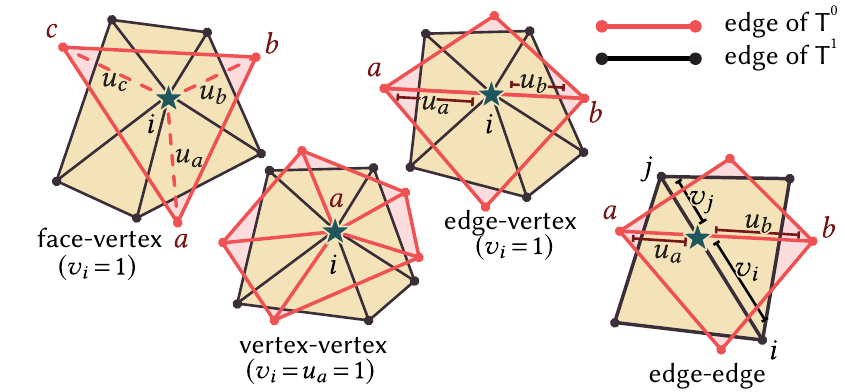}
  \caption{Points are encoded relative to both triangulations \(\T0\) and \(\T1\).  For each triangulation we store the simplex containing the point, and the barycentric coordinates within that simplex.  Here we show a few examples.\label{fig:BarycentricCoordinates}}
\end{figure}

\subsection{Integer Coordinates}
\label{sec:IntegerCoordinates}

\paragraph{Normal coordinates.}
We use normal coordinates to count the number times each edge of \(\T1\) crosses the edges of \(\T0\) (see \figref{IntrinsicViewpoint}, \figloc{left}).
In principle, normal coordinates could either be defined as a value per edge of $\T0$, counting the number of crossings from $\T1$, or as a value per edge of $\T1$.
In our setting, we take the latter approach, as it remains fully-informative even when we insert new vertices in to $\T1$.
It is then natural to think of \(\T1\) as the primary triangulation, with \(\T0\) as a collection of geodesic curves sitting along it (\figref{IntrinsicViewpoint}, \figloc{right}).
To emphasize this abstract viewpoint, we will refer to the edges of $\E0$ as \emph{curves}.
Because we allow edges to be split, a single edge $ab \in \E0$ may actually correspond to a sequence of curves expressed in normal coordinates which meet at intermediate vertices $i \in \Vinserted$.

\begin{figure}
  \includegraphics{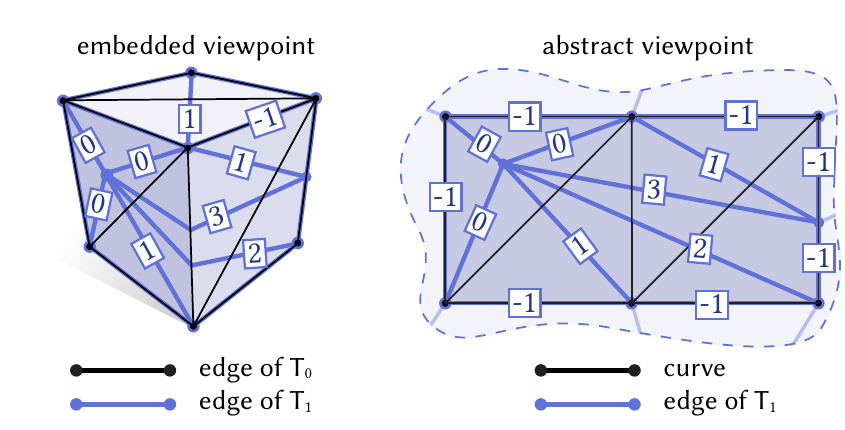}
  \caption{
    The intrinsic triangulation \(\T1\) is often presented embedded in \(\mathbb{R}^3\) on top of the input triangulation \(\T0\) (\figloc{left}). However, in our setting it is helpful to think of \(\T1\) as an abstract intrinsic triangulation, given only by connectivity and edge lengths, which carries a collection of geodesic curves (\figloc{right}).
    \label{fig:IntrinsicViewpoint}
  }
\end{figure}

Precisely, we use \(n: \E1 \to \mathbb{Z}\) to denote the normal coordinates.  For each edge \(ij \in \E1\), the quantity \(n_{ij}\) indicates how many times \(ij\) is crossed by edges of \(\E0\): if \(n_{ij} > 0\), then this is a count of crossings (formally, transversal intersections), whereas if \(n_{ij} = -1\), it indicates that a curve runs along edge $ij$ (\figref{IntrinsicViewpoint}, \figloc{left}).
Note that $n_{ij}$ is never less than $-1$ because we are working with triangulations, and multiple edges of a triangulation cannot lie along the exact same path.
We use $\smash{n^{+}_{ij}}$ to denote the number of transversal crossings, \ie{} \(\smash{n^{+}_{ij}} := \max(n_{ij}, 0)\), and similarly define \(\smash{n^-_{ij}} := -\min(n_{ij}, 0)\), which is 1 on shared edges and 0 otherwise.

\begin{figure}[b]
  \includegraphics{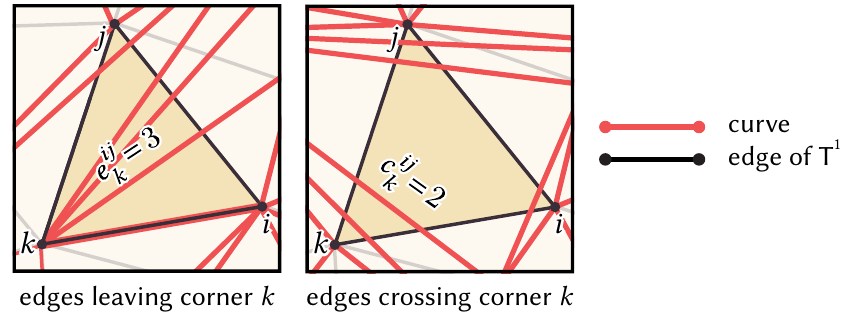}
  \caption{
    It is often useful to count the curves emanating from the interior of each corner (\figloc{left}), and the curves crossing each corner (\figloc{right}).
    \label{fig:CandE}
  }
\end{figure}
Additionally, we define two quantities at each corner of \(\T1\), which count how many curves emanate from the interior of the corner and cross the corner  \resp{} (\figref{CandE}):
\begin{align}
  e_k^{ij} &:= \max(0, n^+_{ij} - n^+_{jk} - n^+_{ki}),  \label{eq:EmanationCount}\\
  c_k^{ij} &:= \tfrac 12 \left( \max\left( 0, n^+_{jk} + n^+_{ki} - n^+_{ij} \right) - e_i^{jk} - e_j^{ki}\right).  \label{eq:CornerCoords}
\end{align}

\paragraph{Crossings.}
A point where a curve crosses an edge \(ij \in \E1\) can be described either in a combinatorial or geometric sense.
A \emph{combinatorial crossing} is given by a pair $\zeta = (\halfedge{ij}\!, p)$, such that \(\zeta\) is the $p^\textrm{th}$ crossing along oriented edge $\halfedge{ij} \in \H1$.
A \emph{geometric crossing} is similarly given by $z = (\halfedge{ij}\!, p, u, v)$, where $u,v$ encode the location of the crossing along the curve and along the $\halfedge{ij}$ (\resp{}) in barycentric coordinates.
Importantly, both of these crossings are \emph{oriented}: the choice of halfedge \(\halfedge{ij}\) versus \(\halfedge{ji}\) indicates which side of the edge one is ``coming from'' and ``going to,'' which will be important when tracing out curves along the surface.  We let $\overline{\zeta} := (\halfedge{ji},n_{ij} - p - 1)$ denote a reversal of orientation.

\paragraph{Roundabouts} Normal coordinates alone do not fully encode the correspondence between \(\T0\) and \(\T1\)\!, because we cannot necessarily determine which curve along $\T1$ corresponds to which edge of $\T0$ (recall that there may be multiple edges of \(\T0\) between the same pair of vertices). Thus we additionally store \emph{roundabouts} \(r: \H1 \to \mathbb{Z}_{\geq 0}\), introduced by \citet[Section 5.2]{Gillespie:2021:DCE}, which describe how the edges of \(\T0\) and \(\T1\) are interleaved around vertices. Unlike Gillespie \etal, we may insert new vertices---however, we still store roundabouts only at halfedges pointing away from shared vertices \(a \in \V0\). These are sufficient to disambiguate the identity of any traced curve, since all edges of \(\E0\) must start and end at vertices in \(\V0\) (if an edge has been split, then the \emph{sequence} of curves starts and ends at vertices in $\V0$).

Precisely, for each halfedge \(\halfedge{aj} \in \H1\) starting at a shared vertex \(a \in \V0\), the roundabout stores the first halfedge \(ab \in \H0\) following \(aj\). This is encoded as an index \(r_{\halfedge{aj}} \in \mathbb{Z}_{\geq 0}\), where we enumerate the halfedges of \(\T0\) about vertex \(a\) in counterclockwise order (\figref{Roundabouts})
\begin{figure}
  \includegraphics{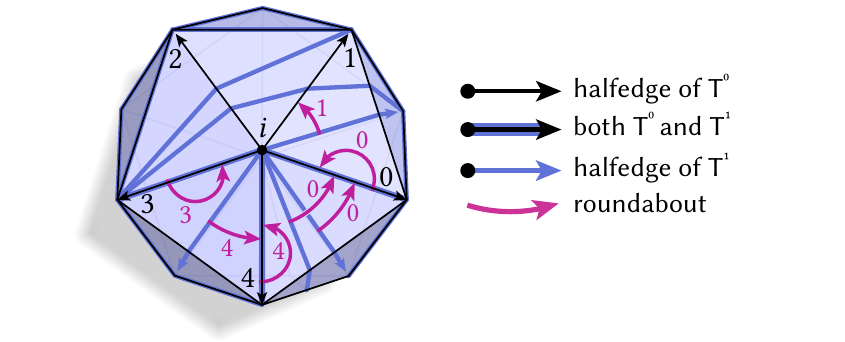}
  \caption{We employ \emph{roundabouts} to encode how edges of \(\T0\) and \(\T1\) are interleaved around vertices.
    \label{fig:Roundabouts}
  }
\end{figure}

\section{Algorithms and Data Structures}
\label{sec:AlgorithmsAndDataStructures}
In this section, we provide descriptions of our data structure and the operations that it supports.
Detailed pseudocode can be found in \appref{Pseudocode}.

\subsection{Data Structure}
\label{sec:DataStructure}

The most essential data to maintain is a mesh of triangulation \(\T1 = (\V1,\E1,\F1)\); recall that \(\V1 \supseteq \V0\).
We use a halfedge mesh, since halfedge meshes can represent general \(\Delta\)-complexes (\secref{Connectivity}), though one could also use a vertex-face adjacency list plus a small amount of additional data~\cite[Section 4.1]{Sharp:2020:LNT}. On top of this mesh, our data structure maintains four quantities:
\begin{itemize}
   \item lengths \(\ell_{ij} \in \mathbb{R}_{>0}\) for each edge \(ij \in \E1\),
   \item normal coordinates \(n_{ij} \in \mathbb{Z}\) for each edge \(ij \in \E1\),
   \item roundabouts \(r_{\he{aj}} \in \mathbb{Z}_{\geq 0}\) for each halfedge \(\halfedge{aj} \in \H1\) incident on a vertex \(a\) shared with \(\T0\),
   \item barycentric coordinates \(\Tpos0i\) relative to \(\T0\) for each \(i \in \V1\).
\end{itemize}

This differs from the scheme of \citet[Section 5]{Gillespie:2021:DCE} in a few key ways. The essential difference is that their data structure assumes that \(\T0\) and \(\T1\) share the same vertex set (\(V^1 = V^0\)). This has numerous consequences---\eg{} they use nonnegative normal coordinates \(n_{ij} \in \mathbb{Z}_{\geq 0}\), they do not store input positions \(\Tpos0i\). Most importantly, it is impossible to perform many of the local mesh operations we describe in the next section without changing the vertex set of \(\T1\); our generalization of the representation is essential if one wants to perform tasks such as Delaunay refinement (\secref{DelaunayRefinement}).

\subsection{Extracting Curves}
\label{sec:ExtractingCurves}

Our first task is to recover a curve on \(\T1\) from its normal coordinates---because our curves are geodesic, we can determine the exact geometry from these normal coordinates.
In particular, we describe a procedure \(\Proc{ExtractCurve}\) (\algref{ExtractCurve}) which takes in any combinatorial crossing $\zeta$ along a curve, and computes the curve's trajectory along \(\T1\) as a sequence $(i, z_1, \dots, z_k, j)$ of geometric crossings along with start and end vertices $i,j \in \V1$.
We note that this mirrors the discussion in \citet[Section 6]{Gillespie:2021:DCE}, albeit in a more general setting; we include a full description here for completeness.

We proceed in two steps, first determining the triangle strip that the curve passes through and only then computing the curve's geometry.
Note that the triangle strip depends solely on the integer-valued normal coordinates $n$, while geometric data and floating point computation are relegated to the second step.

The first step is performed by \Proc{TraceFrom} (\algref{TraceFrom}), which takes some combinatorial crossing \(\zeta = (\halfedge{ij}, p)\) along a curve \(\gamma\), and traces out the remaining combinatorial crossings until \(\gamma\) terminates at a vertex.
\Proc{TraceFrom} proceeds iteratively, taking the crossing where \(\gamma\) enters a triangle, and using the triangle's normal coordinates to determine where \(\gamma\) exits (see \figref{TraceFromCases}).
The direction in which to trace the curve is determined by the orientation of $\halfedge{ij}$.

To determine the triangle strip containing \(\gamma\), \Proc{ExtractCurve} calls \Proc{TraceFrom} once in either direction, yielding the sequence of all combinatorial crossings along \(\gamma\).
\Proc{ExtractCurve} then unfolds this triangle strip in an arbitrary planar coordinate system and draws \(\gamma\) as a straight line between its endpoints.
The intersection of this line with each of the intermediate edges determines the geometric crossings along \(\gamma\) (\figref{TriangleStripLayout}).
Note that unlike \citet[Algorithm 1]{Gillespie:2021:DCE}, we may have to invoke \Proc{ExtractCurve} multiple times to extract a single edge \(ab \in \E0\), as it may pass through several vertices of \(\T1\), \eg{} due to \emph{edge splits} (\secref{EdgeSplit}).

Finally, it is sometimes useful to convert a single combinatorial crossing \(\zeta\) into a geometric crossing \(z\). We will refer to this operation as \Proc{ExtractGeometricCrossing}; it may be implemented by calling \Proc{ExtractCurve} and then returning the single desired crossing.

\begin{figure}
  \includegraphics{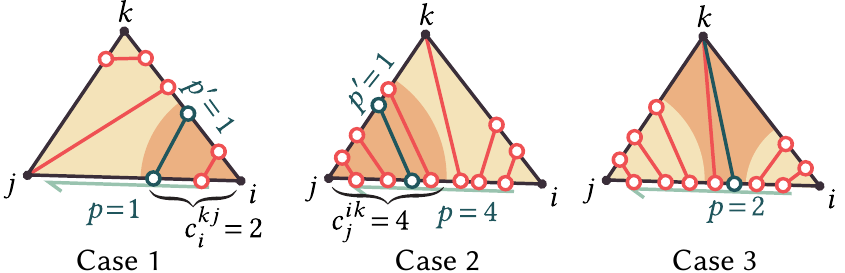}
  \caption{A curve entering triangle \(jik\) along edge \(ij\) can proceed in 3 ways: it can exit along edge \(ik\), in which case it is counted by \(c_i^{kj}\) (\figloc{left}); it can exit along edge \(kj\), in which case it is counted by \(c_j^{ik}\) (\figloc{center}); or it can terminate at vertex \(k\) (\figloc{right}). This forms the core of procedure \Proc{TraceFrom}. \label{fig:TraceFromCases}}
\end{figure}
\begin{figure}
  \includegraphics{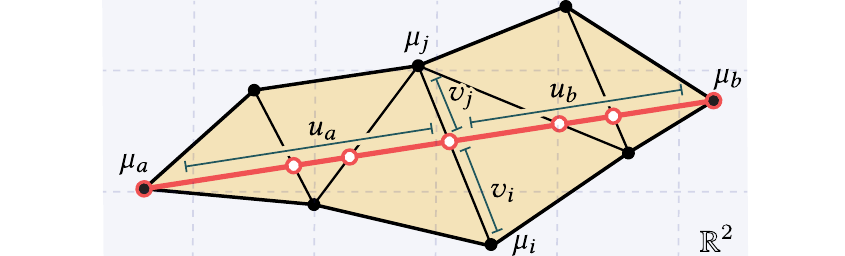}
  \caption{We compute barycentric coordinates by laying out a triangle strip in the plane. \label{fig:TriangleStripLayout}}
\end{figure}

\subsection{Edge Flip}
\label{sec:EdgeFlip}
\begin{figure}
  \includegraphics{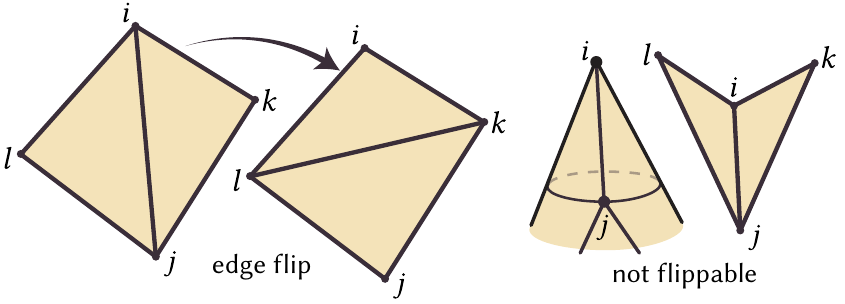}
  \caption{An edge flip replaces an edge with its opposite diagonal (\figloc{left}). An edge \(ij\) is not be flippable if it would leave a vertex with degree zero, or if its neighboring faces form a nonconvex quadrilateral (\figloc{right}). \label{fig:Flippability}}
\end{figure}

Edge flips are a well-studied operation, but we include a discussion here for completeness. We may flip and edge if and only if (i) both endpoints have degree at least one after the flip, and (ii) the two triangles containing the edge form a convex quadrilateral (\figref{Flippability}).

\paragraph{Mesh Update} We replace edge \(ij\) with an edge \(kl\).
We compute the new edge length \(\Tl1_{lk}\) by laying out the two old triangles \(ijk, lji\) in the plane and measuring the length of the appropriate diagonal.

\paragraph{Normal Coordinates \& Roundabouts}
The new normal coordinate \(n_{kl}\) does not depend at all on the geometry of \(\T1\). It is given by the following formula:
\begin{equation}
  \begin{aligned}
    n_{kl} = c_l^{jk} + c_k^{ij}  + &\tfrac12\left|c_j^{il} - c_j^{ki}\right| + \tfrac12\left| c_i^{lj}-c_i^{jk} \right| - \tfrac12 e_l^{ji} - \tfrac12 e_k^{ij} \\
    &\quad           + e_i^{lj} + e_i^{jk} + e_j^{il} + e_j^{ki} +  n_{ij}^- .
\end{aligned}
\end{equation}
This differs slightly from the formula of \citet{Gillespie:2021:DCE}, which did not allow for inserted vertices.

We can update each roundabout from its previous neighbor:
\begin{equation}
  r_{\he{kl}} = \mod\left(r_{\he{ki}}  + e_k^{il} + n_{ki}^-, \deg_0(k)\right),
  \label{eq:RoundaboutUpdate}
\end{equation}
where \(\deg_0(k)\) is the degree of vertex \(k\) in triangulation \(\T0\). The quantity \(\smash{e_k^{il} +n_{ki}^-}\) counts how many edges of \(\T0\) are between \(\halfedge{ki}\) and \(\halfedge{kl}\): \(\smash{e_k^{il}}\) counts edges strictly between them, and \(\smash{n_{ki}^-}\) adds one if there is also an edge lying exactly along \(\halfedge{kl}\).

\vspace{0.25\baselineskip}
\noindent\includegraphics{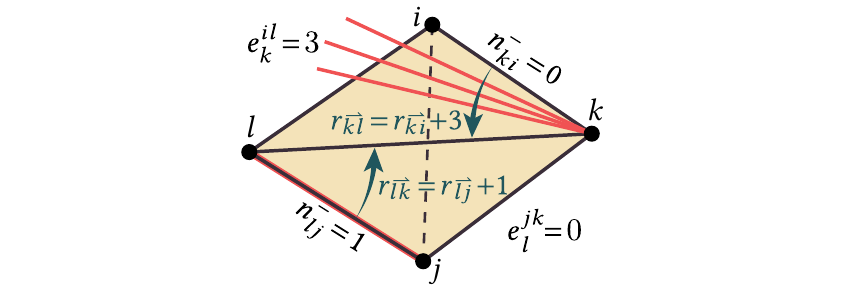}
\vspace{-0.25\baselineskip}
We only perform this update for halfedges whose source is in \(\V0\).

\subsection{Face Split}
\label{sec:FaceSplit}

We now describe procedure \Proc{SplitFace} (\!\algrefs{SplitFace1, SplitFace2}), the first of several new routines to mutate triangulation \(\T1\) while tracking the correspondence with \(\T0\).  Note that \citet[Section 5.4]{Schaefer:2002:ANC} describe a similar face split operation in the topological setting, but do not provide the ability to insert a point at a particular geometric location, which is essential in our setting of Euclidean polyhedra.

In particular, suppose we wish to insert a vertex at a point \(x \in M\), given by barycentric coordinates \(v\) on a triangle \(ijk \in \F1\).
To do so, we need to update the connectivity \(\T1\), edge lengths \(\Tl1\), normal coordinates \(n\), and roundabouts \(r\). Additionally, we need to compute the position \(\Tpos0{}\) of this new vertex in barycentric coordinates on \(\T0\).

\paragraph{Mesh Update}
We update the connectivity of $\T1$ with a new vertex and three new edges and faces.
We compute new edge lengths as a formula of the barycentric coordinates \(u\).
\citet[Section 3.2]{Schindler:2012:Barycentric} show that the length of a displacement vector \(\delta u_i\) in barycentric coordinates is give by
\begin{equation}
  \|\delta u_i\|^2 = -\ell_{ij}^2 \delta u_i \delta u_j - \ell_{jk}^2 \delta u_j \delta u_k - \ell_{ki}^2 \delta u_k \delta u_i.
  \label{eq:NewEdgeLengths}
\end{equation}

\setlength{\columnsep}{.5em}
\setlength{\intextsep}{.5em}
\begin{wrapfigure}{l}{54pt}
  \vspace{-0\baselineskip}\includegraphics{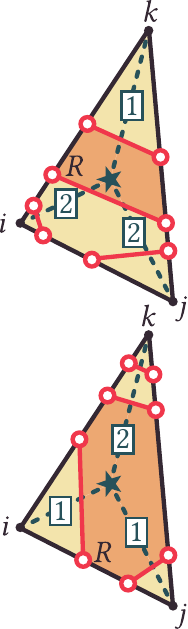}
\end{wrapfigure}
\paragraph{Normal Coordinates \& Roundabouts}
Unlike the case of an edge flip, where the new normal coordinate depends solely on the initial normal coordinates, vertex insertion is an inherently geometric operation.
Different points necessarily result in different normal coordinates (see inset), depending on the region \(R\) in which the point lies.

Concretely, we first compute the geometric crossings of all curves passing through face $ijk$ (using the \Proc{ExtractGeometricCrossing} subroutine).
We then determine which region $R$ the new point lies in via a series of line-side tests.
(One might in principle be able to reduce the number of curves that need to be extracted via lazy evaluation or caching, though we do not pursue such optimizations here.)
We may misclassify points extremely close to a region's boundary due to floating point error, in which case we insert a valid point in the identified region, at a virtually identical location.  Note that this behavior is perfectly reasonable in, \eg{}, retriangulation algorithms (see \secref{DelaunayRefinement}), where the insertion location is not computed exactly anyway.

The roundabouts on any new halfedges emanating from original vertices (\ie, vertices in \(\{i,j,k\} \cap \V0\)) can be set from their neighbors via \eqref{RoundaboutUpdate}.

\paragraph{Position on \(\T0\)}
To determine \(\Tpos0{}\), we must locate the triangle \(abc \in \F0\) containing \(x\), as well as the barycentric coordinates of \(x\) within \(abc\).
We do so via interpolation from the corners of $R$.
Explicitly, the corners of $R$ are all geometric crossings with known barycentric coordinates in some triangle \(abc \in \F0\) (computed in \Proc{ExtractGeometricCrossing}); we can then solve a small linear system to recover the barycentric coordinates of $p$ in the same triangle.
Intuitively, we recover generalized barycentric coordinates for $p$ with respect to the polygon $R$ and apply them on $abc$ to recover standard barycentric coordinates in $\T0$ (see \appref{RecoverBarycentric} for details).

\subsection{Edge Split}
\label{sec:EdgeSplit}
\begin{figure}
  \includegraphics{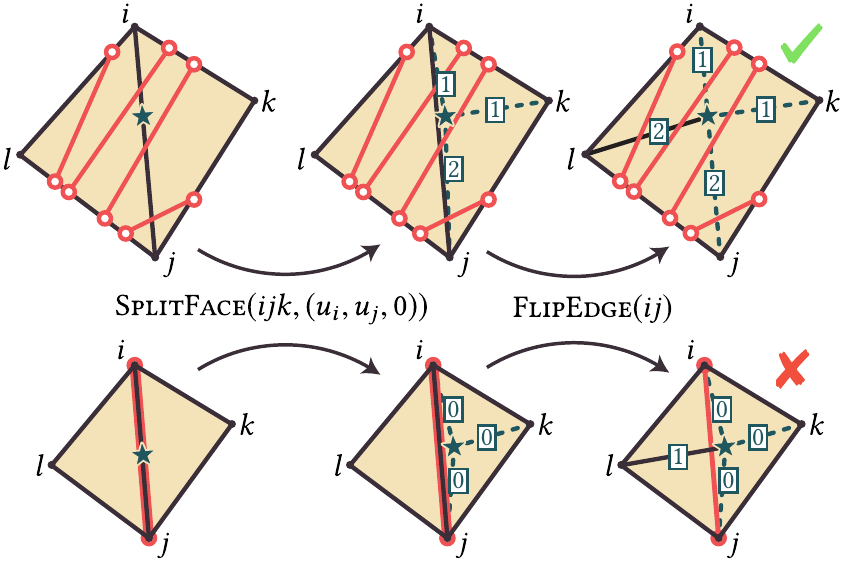}
  \caption{
    Generally, one can split edge \(ij\) by performing a face split on a neighboring face followed by an edge flip (\figloc{top}). However, if \(ij\) carries a curve, this strategy will cause the insertec vertex to miss the curve (\figloc{bottom}). We hence provide a different edge split procedure for this case in \secref{EdgeSplit}.
    \label{fig:FaceSplitEdgeFlip}
  }
\end{figure}

We also introduce an operation \Proc{SplitEdge} (\algref{SplitEdge}), which takes as input a point given by barycentric coordinates \(v\) along an oriented edge \(\halfedge{ij} \in \H1\).
If \(ij\) does not have a curve running along it (\ie{} \(n_{ij} \geq 0\)), then this is implemented as a face split followed by an edge flip (\figref{FaceSplitEdgeFlip}, \figloc{top}).
However, if \(n_{ij} < 0\) (which is common in practice---\eg{} \secref{DelaunayRefinement}), we perform an explicit edge split which inserts the new vertex along the coincident curve (\figref{FaceSplitEdgeFlip}, \figloc{bottom}).

\paragraph{Mesh Update}
We insert a new vertex and triangulate any adjacent faces, computing the new edge lengths via \eqref{NewEdgeLengths}.

\setlength{\columnsep}{.5em}
\setlength{\intextsep}{.5em}
\begin{wrapfigure}{l}{44pt}
  \vspace{-0.5\baselineskip}\includegraphics{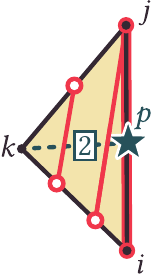}
\end{wrapfigure}
\paragraph{Normal Coordinates \& Roundabouts}
When \(n_{ij} < 0\) the new normal coordinates are simple functions of the old ones, since every curve in face \(ijk\) must emanate from \(i\) or \(j\), or cross \(k\).
The number of such curves is \(\max(n_{ki}, n_{jk}, 0)\) and edge \(pk\) crosses them all.
We hence set \(n_{pj}\) and \(n_{pi}\) equal to \(n_{ij}\), and set
\begin{equation}
  n_{pk} = \max(n_{ki}, n_{jk}, 0).
  \label{eq:NormalCoordinateEdgeSplit}
\end{equation}

As with face splits, roundabouts on any new halfedges emanating from original vertices can be set from their neighbors (\eqref{RoundaboutUpdate}).

\paragraph{Position on \(\T0\)}
To determine \(\Tpos0{}\), we must locate the edge \(ab \in \E0\) containing \(x\), as well as the barycentric coordinates of \(x\) within \(ab\).
Since, \(i\) and \(j\) necessarily have known locations along some edge in $ab \in \E0$, we can simply interpolate by $v$ to compute the location \(\Tpos0{}\).

\subsection{Vertex Removal}
\label{sec:VertexRemoval}

In general, a vertex which is present in the original triangulation cannot be removed without distorting the intrinsic metric because any curvature at that vertex would be lost.  However, inserted vertices \(i \in \Vinserted\) have no curvature, and can hence be removed safely. In fact this operation will be necessary for Delaunay refinement of domains with boundary (\secref{DelaunayRefinement}).

The basic strategy behind \Proc{RemoveVertex} (\algref{RemoveVertex}) is to flip edges incident on the vertex to be removed until it has degree three, then delete the three edges incident on the vertex as well as the vertex itself.
No other data needs to be updated, since the edges of the resulting triangle already appear in the triangulation.
\algref{RemoveVertex} describes this procedure, and \thmref{VertexRemovalFlipping} proves its correctness for simplicial complexes.
A nearly identical procedure can be used to remove an inserted boundary vertex.
\citet[Section 5.4]{Schaefer:2002:ANC} also suggest a similar flipping procedure, but work in the topological setting where the necessary edge flips are always valid---they do not consider the convexity condition (\secref{EdgeFlip}).

\subsection{Moving Inserted Vertices}
\label{sec:MovingInsertedVertices}

Given the previous operations, we can easily define a procedure for moving around inserted vertices. Specifically, given a vector \(v\) in the tangent space of an inserted vertex \(i\), we can move \(i\) along \(v\) in the following way:
\begin{itemize}
   \item First, compute the new location \(p = \exp_i(v)\).
   \item Insert \(p\) using \Proc{SplitFace}.
   \item Remove \(i\) using \Proc{RemoveVertex}.
\end{itemize}
We insert \(p\) first since the removal procedure could flip edges incident on the triangle containing \(p\), invalidating its barycentric coordinates.  Note that Sharp \etal{} propose an alternative strategy for local vertex displacement~\cite[Section 3.3.3]{Sharp:2019:NIT}.

\subsection{Common Subdivision}

\label{sec:CommonSubdivision}
As noted previously, the \emph{common subdivision} $S$ of $\T0$ and $\T1$ is the polygon mesh obtained by ``slicing up'' the underlying surface along the edges of both \(\T0\) and \(\T1\).  The vertices of \(S\) are hence a superset of \(\V0\) and \(\V1\), and every edge or face of \(\T0\) and \(\T1\) can be expressed as union of edges or faces of \(S\) (\resp).  Moreover, the faces of \(S\) are always planar and convex.  Most importantly in our setting, any piecewise-linear function on $\T0$ or $\T1$ can be represented exactly as a piecewise-linear function on $S$.  Note however that even if \(\T0\) and \(\T1\) have nice elements, $S$ is not in general a high-quality mesh, and may not itself be suitable for, \eg{}, solving PDEs.  Rather, it plays a complementary role in the geometry processing pipeline, enabling (for instance) transfer of data between triangulations (\secref{AttributeTransfer}), or visualization of data downstream via standard rendering tools.

\begin{figure}
  \begin{center}
  \end{center}
  \includegraphics{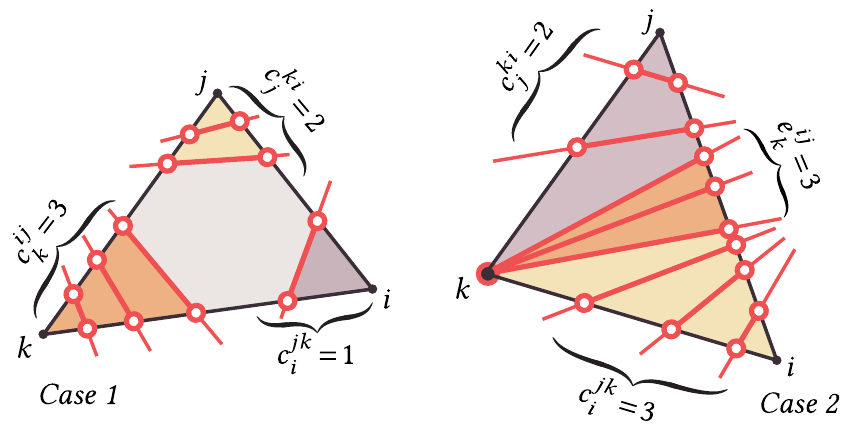}
  \caption{
    We extract the connectivity of common subdivision within each triangle using its normal coordinates.
    \label{fig:CommonSubdivision}
  }
\end{figure}

We compute the common subdivision by cutting \(\T1\) along the edges of \(\T0\).
First we extract the connectivity of \(S\), using only the normal coordinates \(n_{ij}\).
Then we recover the intersection geometry, allowing us to interpolate data stored at the vertices of \(\T0\) or \(\T1\) to \(S\)---most commonly, vertex positions on \(\T0\) along with any solution data on \(\T1\).
Note that this procedure was previously described by \citet[Section 3.4.2]{Sharp:2019:NIT}; we recap it here for completeness, and to give a convenient description using our integer coordinates.

\paragraph{Connectivity.}
We subdivide \(\T1\) independently in each face \(ijk\).
The normal coordinates alone determine the connectivity of \(S\) within this face.
There are just two cases to consider, illustrated in \figref{CommonSubdivision}.
Case 1 occurs when no curves emanate from any corner, so we simply need to connect the first \(\smash{c_i^{jk}}\) crossings along edge \(ij\) to the first \(\smash{c_i^{jk}}\) crossings along \(ik\) (in order), and likewise for corners \(j\) and \(k\).
In Case 2 curves emanate from some corner; without loss of generality, let this corner be \(k\) so that the number of such curves is \(\smash{e_k^{ij}} > 0\) and there are more curves crossing edge \(ij\) than the other two edges.
Hence, we can walk from \(i\) to \(j\), connecting the first \(\smash{c_i^{jk}}\) crossings to those along \(ik\), the next \(\smash{e_i^{jk}}\) crossings to vertex \(k\), and the remaining \(\smash{c_j^{ki}}\) crossings to those along edge \(kj\).
Note that curves running along edges (\(n_{ij} < 0\)) require no special treatment.

\paragraph{Intersection Geometry.} Next, we associate each vertex \(i\) of the common subdivision with a point in \(\T0\) and a point in \(\T1\), encoded in barycentric coordinates relative to some simplex (vertex, edge, or face) of the appropriate triangulation.
Using this, one can linearly interpolate data in the usual way.
Again, there are just two cases: each vertex \(i\) in \(S\) is either a vertex of \(\T1\) or the intersection of an edge of \(\T0\) with an edge of \(\T1\).
In the first case, the position on \(\T1\) is given by \(i\) itself, and its position \(\Tpos0i\) on \(\T0\) was computed when \(i\) was inserted.
In the second case, we compute the desired barycentric coordinates using \Proc{ExtractCurve} (see \algref{ComputeCommonSubdivision} for details).

\subsection{Transposing Coordinates}
\label{sec:TransposingCoordinates}

Throughout, we store normal coordinates \(n:\E1 \to \mathbb{Z}\) which count how many times edges of \(\E1\) cross edges of \(\E0\).
It is sometimes useful to observe that tracing an edge \(ab \in \E0\) over \(\T1\) counts how many times \(ab\) crosses edges of \(\T1\).
If \(\T1\) has more vertices than \(\T0\), then the edges of \(\E1\) are not \emph{normal} over \(\T0\)---they can start and end in the middle of faces of \(\T0\)---and these crossing counts do not uniquely encode the structure of \(\T1\).
However, if \(\T0\) and \(\T1\) do have the same vertex set, then these crossing counts provide an implicit representation of \(\T1\) as a collection of curves over \(\T0\).

\subsection{Visualization}
\label{sec:Visualization}

In the following, we show examples of intrinsic triangulations on top of meshes (\eg{} \figref{IDTR}).
To produce these figures, we compute the common subdivision (\secref{CommonSubdivision}) and draw the edges of the input mesh with a black wireframe while coloring the intrinsic triangles in arbitrarily-chosen colors.
In figures displaying functions defined on intrinsic triangulations (\eg{} \figref{HeatMethod}), we interpolate the solutions along the common subdivision for rendering.

\subsection{Robust Implementation}
\label{sec:RobustImplementation}

Our integer coordinates are guaranteed to encode a triangulation sitting atop \(\T1\).
The geometric accuracy of this triangulation, of course, depends on floating point arithmetic, which can become inaccurate in near-degenerate configurations.
\emph{Exact predicates} have been applied with great success to similar problems~\cite{Devillers:2003:EEG}.
Unfortunately they do not directly apply to intrinsic triangulations, as the predicates that we evaluate are not fixed functions of the input data; an intrinsic edge length can depend upon arbitrarily many input edge lengths.
Hence, we focus on fast and robust implementations using ordinary floating point arithmetic.

One essential tool for dealing with intrinsic triangulations on near-degenerate input meshes is \emph{intrinsic mollification}, introduced by \citet{Sharp:2020:LNT}. Mollification improves degenerate meshes by adding a small \(\epsilon\) to every edge length, provably improving triangle quality. This changes the geometry by a negligible amount, and moreover we only mollify if some triangle is within \(\epsilon\) of being degenerate. This procedure works particularly well with our data structure compared to signposts: the signpost data structure relies on tracing queries along the surface which become \emph{less} accurate when mollification is applied. Our integer coordinates have no such problem: we always get the correct edge sequence, even if the mesh geometry is slightly modified.  In our experiments we mollify with \(\epsilon = 10^{-5}\), and find that it resolves almost all numerical difficulties.

Even after mollification, it is still beneficial to use care when working with floating point. For example, there are well-conditioned triangles on which the Delaunay condition (\eqref{DelaunayCotans}, discussed in the next section) is difficult to evaluate; in practice, we only enforce \eqref{DelaunayCotans} up to some \(\epsilon\) tolerance. As a further example, when computing new normal coordinates in \Proc{SplitFace}, one could lay out the face in the plane, and independently count intersections along the new edges. However, this can produce invalid normal coordinates in floating point. We apply a more complicated policy (see \appref{Pseudocode}) which always yields valid normal coordinates. For additional details on all procedures, we refer the reader to our implementation, which will be made available after review.

\subsection{Other Algorithms}
\label{sec:OtherAlgorithms}

Normal coordinates also enable a wide variety of other operations not detailed here.  For instance, \citet[Section 5]{Schaefer:2002:ANC} provide algorithms for counting connected components, checking if crossings are part of the same curve, checking if curves are isotopic, and computing the oriented intersection number.  \citet{Erickson:2013:TCC} provide an asymptotically-fast algorithm for tracing normal curves across a surface.  Finally, \citet[Proposition 13]{Dynnikov:2020:CINC} provides an algorithm for computing how many times curves represented by normal coordinates intersect.

\section{Applications}
\label{sec:Applications}

\subsection{Intrinsic Delaunay Triangulations}
\label{sec:IntrinsicDelaunay}

\begin{figure}
\begin{center}
\end{center}
    \includegraphics[width=\columnwidth]{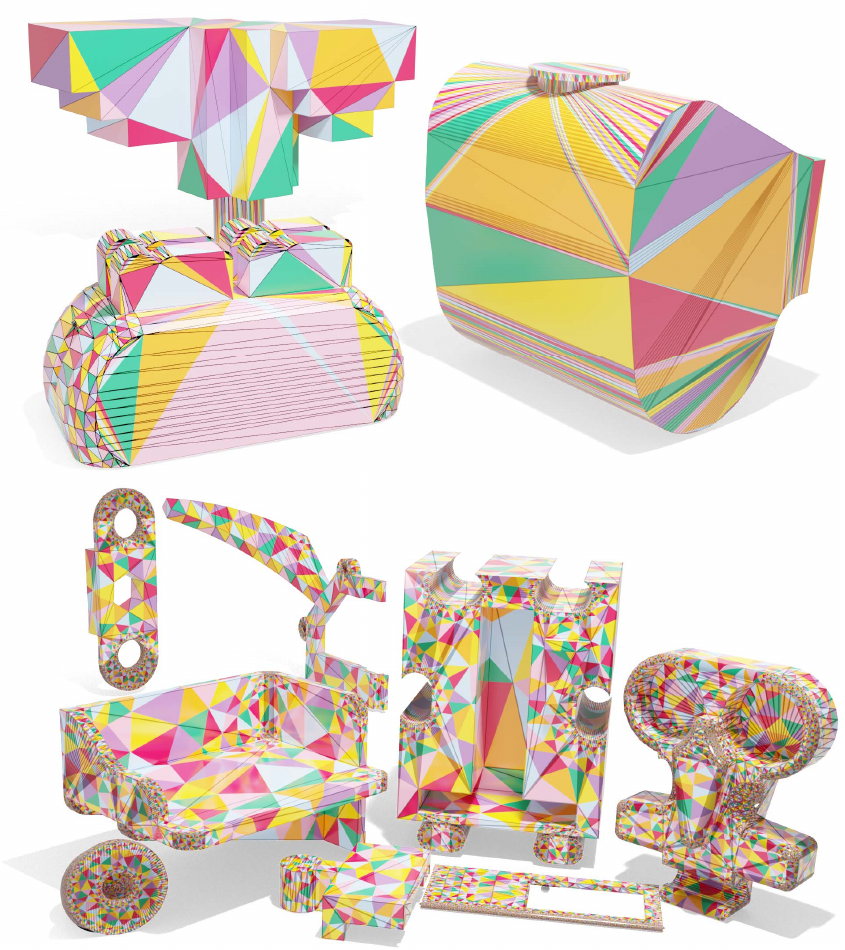}
    \caption{
      Using our integer-based data structure, we can not only improve near-degenerate meshes by generating \emph{intrinsic Delaunay triangulations} (\figloc{top}), but can also extract the common subdivision after computing a high-quality \emph{intrinsic Delaunay refinement} (\figloc{bottom}).
      \label{fig:IDTR}
    }
\end{figure}
\begin{wrapfigure}[4]{r}{70pt}
  \vspace{-1.5\baselineskip}\includegraphics{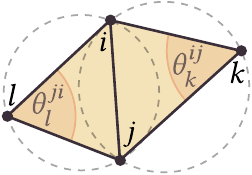}
\end{wrapfigure}
One key application of intrinsic triangulations is the computation of \emph{intrinsic Delaunay triangulations} (\figref{IDTR}, \figloc{top}).
A triangulation is said to be Delaunay if the sum of angles opposite every edge is at most \(\pi\), \ie{} for every \(ij \in E\) we have
\begin{equation}
    \theta_k^{ij} + \theta_{l}^{ji} \leq \pi.
    \label{eq:DelaunayCondition}
\end{equation}

Delaunay triangulation have a number of beneficial properties.
One consequence of \eqref{DelaunayCondition} is that edges of a Delaunay triangulation must have nonnegative \emph{cotan weights}:
\begin{equation}
  \cot \theta_k^{ij} + \cot \theta_l^{ji} \geq 0.
  \label{eq:DelaunayCotans}
\end{equation}
In fact, \eqref{DelaunayCotans} is equivalent to \eqref{DelaunayCondition} above, and provides a convenient formula for checking the Delaunay property in an intrinsic triangulation. Moreover, \eqref{DelaunayCotans} ensures that the finite element Laplacian \(L\) satisfies the \emph{maximum principle}, guaranteeing that discrete harmonic functions do not have local extrema in the interior of the domain \cite[Proposition 19]{Bobenko:2007:ADL}.
Similarly \eqref{DelaunayCotans} also ensures that discrete harmonic vector fields are ``flip-free'' \cite[Section 5.4]{Sharp:2019:VHM}.
Furthermore, the local Delaunay condition implies the \emph{empty circumcircle property}: each triangle's geodesic circumdisk contains no vertices, illustrated in \figref{Circumcenter}, \figloc{left} \cite[Proposition 10]{Bobenko:2007:ADL}.

The Delaunay triangulation can be computed via a simple greedy algorithm: flip any non-Delaunay edge until all edges satisfy \eqref{DelaunayCondition} \cite[Propositions 11 and 12]{Bobenko:2007:ADL}.

\subsection{Intrinsic Delaunay Refinement}
\label{sec:DelaunayRefinement}

Delaunay refinement inserts vertices in order to produce a Delaunay mesh whose triangles all satisfy a minimum angle bound (\figref{IDTR}, \figloc{bottom}).
Here we modify \emph{Chew's second algorithm} to perform intrinsic Delaunay refinement~\cite{Chew:1993:GQMG, Shewchuk:1997:DRMG}.
This problem has been extensively studied in the plane, but an intrinsic (\ie{} geodesic) scheme was only recently proposed by \citet[Section 4.2]{Sharp:2019:NIT}.
However, they did not handle meshes with boundary---here we resolve the essential difficulties of the boundary case, and show how refinement can be implemented using our integer-based data structure.

In the plane, the basic algorithm is to greedily pick any triangle which violates the minimum angle bound, insert a vertex at its circumcenter, then flip to Delaunay.
This process continues until all triangles satisfy the angle bound.
If a triangle's circumcenter is outside the domain, then the boundary edge \(ij\) separating the triangle from its circumcenter is split at its midpoint; subsequently, all interior vertices within at least a distance of \(\ell_{ij}/2\) are removed---though removing additional interior vertices causes no issues (\appref{RemovingExtraVertices}).
One can prove that this process succeeds for minimum angle bounds up to 25.65 degrees on planar domains with boundary angles at least \(60^\circ\)~\cite[Section 3.4.2]{Shewchuk:1997:DRMG}.
More advanced versions of this procedure can achieve better angle bounds, \eg{} \cite{Rand:2011:WWC}, but here we restrict our attention to the basic algorithm for simplicity.

\begin{figure}
  \includegraphics{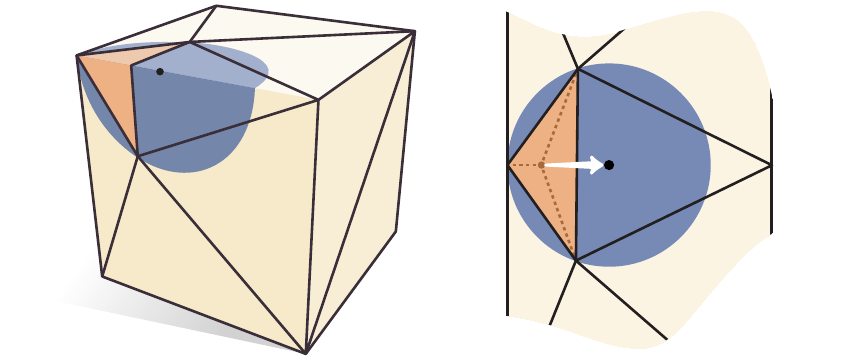}
  \caption{
    Triangles in Delaunay meshes have empty circumdisks, and thus well-defined circumcenters (\figloc{left}). When necessary, we locate a triangle's circumcenter by walking outwards from its barycenter (\figloc{right}).
    \label{fig:Circumcenter}
  }
\end{figure}
There are two difficulties in adapting this algorithm to the intrinsic setting: locating circumcenters and computing (geodesic) distances.
As mentioned earlier, intrinsic Delaunay triangulations obey the empty circumcircle property; hence each triangle has an intrinsically-flat circumdisk with a well-defined center (\figref{Circumcenter}, \figloc{left}).
So long as this center corresponds to a point on the surface, it can be found by walking from the triangle's barycenter (\figref{Circumcenter}, \figloc{right}).
In practice, we compute triangle \(ijk\)'s circumcenter in homogeneous (\ie, unnormalized) barycentric coordinates \(\hat v_i\) via the following formula \cite[Section 2.3]{Schindler:2012:Barycentric}:
\begin{equation}
  \hat v_i := \ell_{jk}^2 (\ell_{ij}^2 + \ell_{ki}^2 - \ell_{jk}^2),
  \label{eq:HomogeneousCircumcenterBary}
\end{equation}
and then normalize to obtain barycentric coordinates
\begin{equation}
  v_i := \tfrac{\hat v_i} {\hat v_i + \hat v_j + \hat v_k}.
  \label{eq:CircumcenterBary}
\end{equation}
To locate the circumcenter on the surface, we then evaluate the exponential map (\secref{Geometry}) starting at the barycenter \(w_i = w_j = w_k = 1/3\), along the vector \(v-w\).  If we hit a boundary edge \(ij\) while tracing out this path, then the circumcenter is not contained in the surface, so we split \(ij\) at its midpoint and flip to Delaunay.  We must then remove all inserted interior vertices within a geodesic ball of radius \(\ell_{ij}/2\) centered at the inserted point.  Computing geodesic distance on a surface mesh is nontrivial, but \citet[Corollary 1]{Xia:2013:TSF} shows that on a Delaunay triangulation any vertex inside a geodesic ball of radius \(r\) will also be inside the Dijkstra ball of radius \(2r\) (\ie{} points whose distance along the edge graph are at most \(2r\)).  We hence remove all interior inserted vertices within a Dijkstra distance of \(\ell_{ij}\).  Note that while Xia considers only the planar setting, their proof (which is based on triangle strips) applies without modification to intrinsic Delaunay triangulations of surfaces.

\setlength{\columnsep}{.5em}
\setlength{\intextsep}{.5em}
\begin{wrapfigure}{r}{40pt}
  \vspace{-0.5\baselineskip}\includegraphics{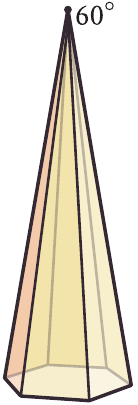}
\end{wrapfigure}
On meshes with narrow cone vertices or boundary angles, it may be impossible to find any triangulation satisfying a given angle bound. In such cases, we do not insert circumcenters of intrinsic triangles which are incident on exactly one narrow vertex, or are entirely contained in a triangle of \(\T0\) which is incident a narrow vertex, and ignore such triangles when computing the minimum corner angle of the output mesh. Although the final output may violate the angle bound, such triangles appear only near narrow vertices. In analogy with the planar case, we set \(60^\circ\) as the minimum allowed angle sum (see inset); in practice the vast majority of meshes obey this constraint at all vertices (97.2\% of Thingi10k), and even on those which do not we obtain high-quality triangulations.

\subsection{PDE-Based Geometry Processing}
\label{sec:PDEBasedGeometryProcessing}

\begin{figure}
  \begin{center}
  \end{center}
  \includegraphics{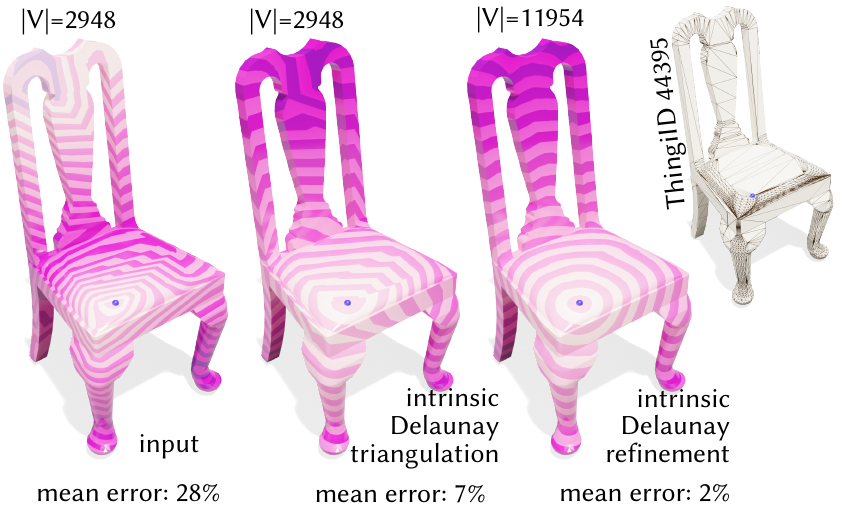}
  \caption{
    Running PDE-based algorithms such as the heat method on poor triangulations (\figloc{left}) can lead to inaccurate solutions. Flipping to intrinsic Delaunay (\figloc{center}) and performing Delaunay refinement (\figloc{right}) can drastically improve the results.
    \label{fig:HeatMethod}
  }
\end{figure}

PDE-based methods abound in geometry processing, as they are generally simple to implement and benefit from decades of research into linear solvers,
Many such methods depend only on intrinsic data, and are hence a natural application of our intrinsic Delaunay triangulations and refinements.
On near-degenerate inputs, simply running the standard algorithm on an intrinsic triangulation instead of the original mesh yields solutions of dramatically higher quality (\figrefs{HeatMethod, LogMapVectorField}).

We show several examples: fast geodesic distance computation \cite{Crane:2017:HMD}, local parameterization via the logarithmic map \cite{Sharp:2019:VHM}, and smooth vector fields \cite{Knoppel:2013:GOD}.
Further examples on tasks such as parameterization, minimal surfaces, and surface editing can be found in \cite{Sharp:2019:NIT,Sharp:2020:LNT}.
Across the board, normal coordinates offer improved robustness guarantees, and open doors to higher solution accuracy with the common subdivision.

\begin{figure}
  \begin{center}
  \end{center}
  \includegraphics{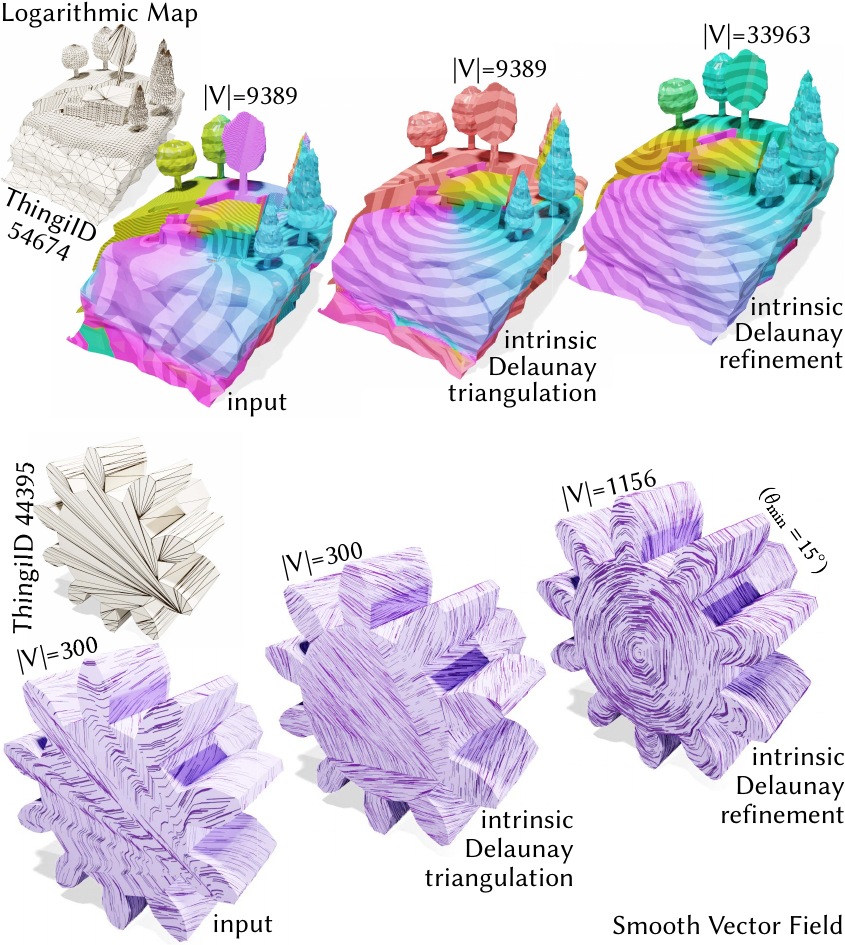}
  \caption{
    Here we compute a local parameterization (the \emph{logarithmic map}, \figloc{top}), and a smooth vector field (\figloc{bottom}) using the connection Laplacian.
    Both procedures yield inaccurate results on near-degenerate inputs (\figloc{left})---intrinsic Delaunay triangulations (\figloc{center}) and intrinsic Delaunay refinements (\figloc{right}) greatly improve solution quality.
    Whether our solution is a scalar function or vector field, we can visualize it on the common refinement.
  \label{fig:LogMapVectorField}
}
\end{figure}

\subsection{Attribute Transfer}
\label{sec:AttributeTransfer}

\begin{figure}
\begin{center}
\end{center}
    \includegraphics{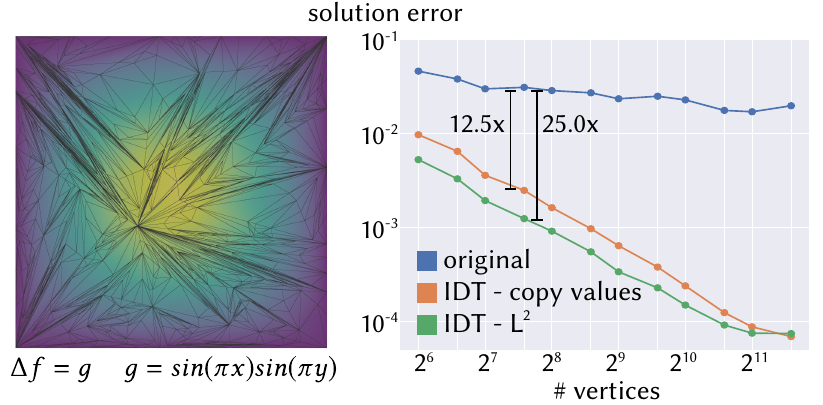}
    \caption{
      Accuracy is improved by transferring PDE solutions back to an original triangulation as the $L^2$-nearest solution, evaluated via the common subdivision.
      Here we generate random low-quality meshes of the unit square by random edge splits (\figloc{left}), and plot the error in the solution of a Poisson equation compared to analytic ground truth, always represented in the basis of the original triangulation (\figloc{right}).
      Each data point is the average error over 100 trials.
      As expected, solving on the intrinsic Delaunay triangulation dramatically increases accuracy, but further improvements are gained by choosing the solution on the original mesh which is $L^2$-nearest to the intrinsic solution, rather than naively copying vertex values.
      \label{fig:AttributeTransferAccuracy}
    }
\end{figure}

Intrinsic triangulations can drastically improve the quality of solutions to PDEs on low-quality meshes.
However in practice, one often needs to represent the solution on the input mesh.
Past approaches have simply ``copied back'' the solution values at vertices of the original mesh, but this strategy is ad-hoc and suboptimal.
A more principled approach is to choose the function on the original mesh which is closest to the intrinsic solution.
Here, we restrict our treatment to piecewise-linear bases and $L^2$ distance for simplicity, though the same strategy could easily be applied to other basis functions and notions of distance.
Precisely, given a function $f$ on the intrinsic triangulation, we seek $\hat f$ on the original mesh that minimizes the squared $L^2$ distance
\begin{equation}
  \|f - \hat{f}\|^2_{L^2} := \int_M |f(x) - \hat{f}(x)|^2 dx.
\end{equation}
Here, $f$ and $\hat f$ are functions represented in finite-dimensional bases with nodal values at the vertices of the intrinsic triangulation and the original mesh \resp{}
In traditional finite elements, this integral commonly arises over a \emph{single} triangulation, in which case it can be evaluated via the Galerkin mass matrx $\Msf$ as 
\begin{equation}
  \|f - \hat{f}\|^2_{L^2} = (f - \hat{f})^{T} \Msf (f - \hat{f}),
\end{equation}
where $\Msf$ is constructed as in \cite[Chapter 10, (32)]{Strang:2008:AFEM}.
However, in our setting $f$ and $\hat{f}$ are encoded over \emph{different} triangulations; they are members of different function spaces.
Our key observation is that the common subdivision $S$ (\secref{CommonSubdivision}) provides exactly the structure needed to evaluate $\|f - \hat{f}\|^2_{L^2}$, as both functions are linear on each triangle of $S$.
In fact, we have
\begin{equation}
  \label{eq:l2solve}
  \|f - \hat{f}\|^2_{L^2} = (P_1 f  - P_0 \hat{f})^{T} \Msf_S (P_1 f - P_0 \hat{f}),
\end{equation}
where now $\Msf_S$ is the Galerkin mass matrix of the common subdivision, and $P_0,P_1$ are interpolation matrices which map piecewise-linear functions on  original and intrinsic triangulations to piecewise-linear functions on \(S\), \resp{}
In particular, $P_0$ is a $|\V0| \times |\V{S}|$ matrix, where each row corresponds to a vertex of $S$, and has that vertex's barycentric coordinates on $\T0$ as entries. 
$P_1$ is defined likewise for $\T1$.
We then find the function $\hat f$ which minimizes \eqref{l2solve} as the solution to a linear least-squares system, which can be prefactored if desired to efficiently transfer many functions.

We can leverage this formulation to transfer functions from any intrinsic triangulation back to the original mesh.
In \figref{AttributeTransferAccuracy}, we show how this transfer indeed improves the accuracy of PDE solutions as measured on the original low-quality mesh.
This machinery is enabled because our integer coordinates efficiently and robustly compute the common subdivision.
More broadly, this paradigm opens the door to a wide variety of future finite-element formulations involving intrinsic triangulations.

\subsection{Flip-Based Geodesic Paths}
\label{sec:FlipBasedGeodesicPaths}

\begin{figure}
\begin{center}
\end{center}
    \includegraphics{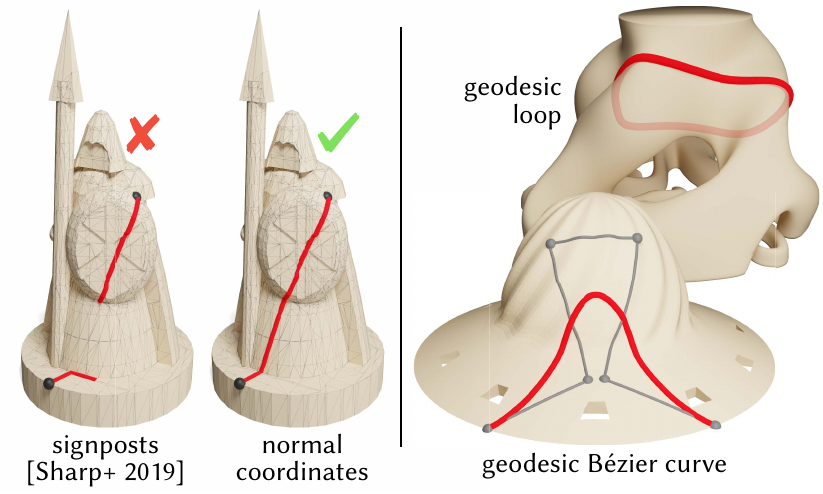}
    \caption{
      We construct geodesic paths by flipping edges in a normal coordinate intrinsic triangulation, as in \cite{Sharp:2020:YCF}.
      Normal coordinates guarantee a valid path, even on degenerate inputs (\figloc{left}).
      This unlocks advanced applications of geodesics in normal coordinates with the same guarantees, such as geodesic loops and B\'{e}zier curves (\figloc{right}).
      \label{fig:FlipGeodesics}
    }
\end{figure}

The previous sections have demonstrated the value of intrinsic triangulations as a high-quality basis for discretizing functions on surfaces; more broadly, these triangulations also provide simple and robust solutions to other tasks across geometry processing.
As an example, the recent \Proc{FlipOut} procedure of \citet{Sharp:2020:YCF} computes exact geodesic paths on surfaces via a simple intrinsic edge flipping strategy, introducing the geodesic as a path of edges in the triangulation.
This method is easily implemented in our integer representation in terms of the mesh operations in \secref{AlgorithmsAndDataStructures}, and the resulting geodesic paths may then be recovered with the \Proc{ExtractEdge} subroutine.
Computing geodesics with our robust integer coordinates is particularly appealing, because geodesic algorithms are notoriously difficult to implement robustly~\cite[Section 5.3]{Sharp:2020:YCF}.
Even the method of Sharp \etal{} uses the signpost data structure, which may fail to reconstruct a connected path along the surface for degenerate inputs.
In contrast, implementing \Proc{FlipOut} in our integer coordinate representation extends the benefits of our approach to this task, including a guarantee of valid connectivity in the output (\figref{FlipGeodesics}, \figloc{left}).
It also enables higher-level tasks involving geodesic paths to be safely run on low-quality input, such constructing geodesic loops on surfaces, and even geodesic B\'{e}zier curves, using a de Casteljau-style scheme due to \citet{Morera:2008:MTG} as shown in \figref{FlipGeodesics}, \figloc{right}.
 
\section{Evaluation}
\label{sec:Evaluation}

\begin{table}[]
\centering
\begin{tabular}{@{}rcr@{}}
\toprule
\textbf{Method}                         & \textbf{\shortstack{Intrinsic\\Delaunay\\Triangulation}}    & \textbf{\shortstack{Intrinsic\\Delaunay\\Refinement}} \\
\midrule
Explicit Overlay                        &  100 \%          & - \\
Signpost Tracing                        & 96.0 \%          & 69.1 \% \\ %
Integer Coordinates                     &  100 \%          & 100 \% \\
\bottomrule
\end{tabular}
\vspace{0.5em}
\caption{
  The success rate of our method and past approaches for building high-quality intrinsic triangulations in the Thingi10k dataset. For each we construct a Delaunay triangulation, either on the original vertex set or with Delaunay refinement to a $25^\circ$ minimum angle bound, and attempt to recover the connectivity of the common subdivision. The explicit overlay method does not support refinement.
  \label{tab:success_rates_thingi10k}
}
\vspace{-.5em}
\end{table}

We implemented all algorithms in C++; since basic vertex-face adjacency list cannot represent a general \(\Delta\)-complex (\secref{Connectivity}), we use a
halfedge data structure for triangle meshes. Timings are measured on a single core of an Intel i9-9980XE with 32 GB of RAM.

\paragraph{Performance}
Generally our data structure is quite fast, computing Delaunay refinements for complex meshes in seconds. For example, computing the Delaunay refinement \figref{HeatMethod} takes 0.2s, and the Delaunay refinement in \figref{LogMapVectorField} (\figloc{top}) takes 0.6s. Because we lazily recover intersection geometry from our integer coordinates when inserting vertices, routines such as Delaunay refinement which perform many insertions may become moderately expensive on large near-degenerate inputs. For instance we take 4 minutes to perform Delaunay refinement on 719791 (\figref{BadMesh}, \figloc{top}) which signposts does in 1.5 minutes, but on such meshes signposts generally fails to compute a valid common subdivision at the end. \secref{FutureWork} discusses hybrid routines which may give the best of both worlds.

\subsection{Robustness}
\label{sec:Robustness}

\begin{figure}
\begin{center}
\end{center}
    \includegraphics{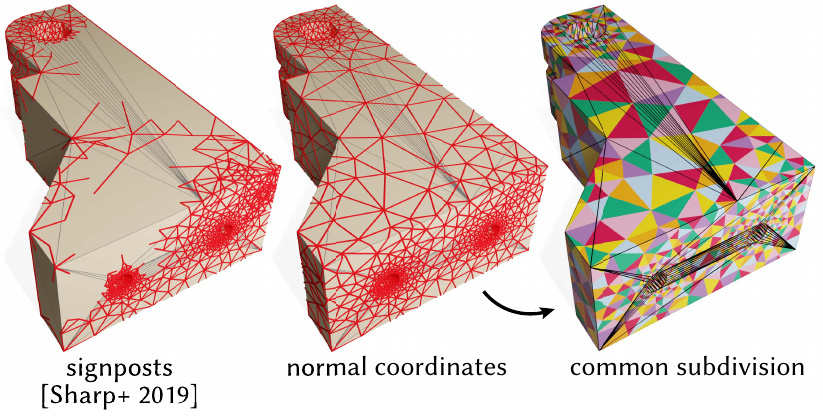}
    \caption{
      Past methods extracted edges by tracing ``signposts'' along the mesh, which may fail in the presence of degenerate triangles.
      In contrast, our integer coordinates always yield a topologically-valid common subdivision, even on extremely poor quality inputs.
      \label{fig:SignpostFail}
    }
\end{figure}

We validate robustness by successfully computing Delaunay triangulations, refinements, and their common subdivisions on all manifold meshes in Thingi10k \cite{Zhou:2016:TTK}. In particular, we used \emph{MeshLab} to convert each mesh to the PLY file format \cite{Cignoni:2008:ML}, resulting in 7696 valid manifold meshes. We begin by mollifying each mesh to a tolerance of \(10^{-5}\) (\secref{RobustImplementation}). For each model we compute the intrinsic Delaunay triangulation (\secref{IntrinsicDelaunay}) with a tolerance of \(10^{-5}\), as well as an intrinsic Delaunay refinement (\secref{DelaunayRefinement}) with a \(25^\circ\) angle bound. We verify that the algorithms terminate with the expected conditions. Additionally, we successfully extract an explicit mesh of the common subdivision in both cases, except for 1 model in the case of refinement whose common subdivision contains around 30 million vertices (\figref{BadMesh}, \figloc{top}).

We compare against the explicit overlay representation of \citet{Fisher:2006:IDT} and the signpost representation of \citet{Sharp:2019:NIT} (\tabref{success_rates_thingi10k}). The overlay representation similarly offers a guarantee of valid connectivity, but does not provide a constant-time edge flip operation (like normal coordinates do).  More importantly it does not support operations beyond edge flips and thus cannot perform Delaunay refinement. Signposts support a wide range of operations, but may not successfully recover the common subdivision on degenerate inputs (\figref{SignpostFail}). The statistic reported here differs from the result in \citet{Sharp:2019:NIT}, because no preprocessing of meshes is performed. For refinement \citet{Sharp:2019:NIT} do not treat the boundary case, so we evaluate only on models without boundary.

\begin{figure}
  \includegraphics{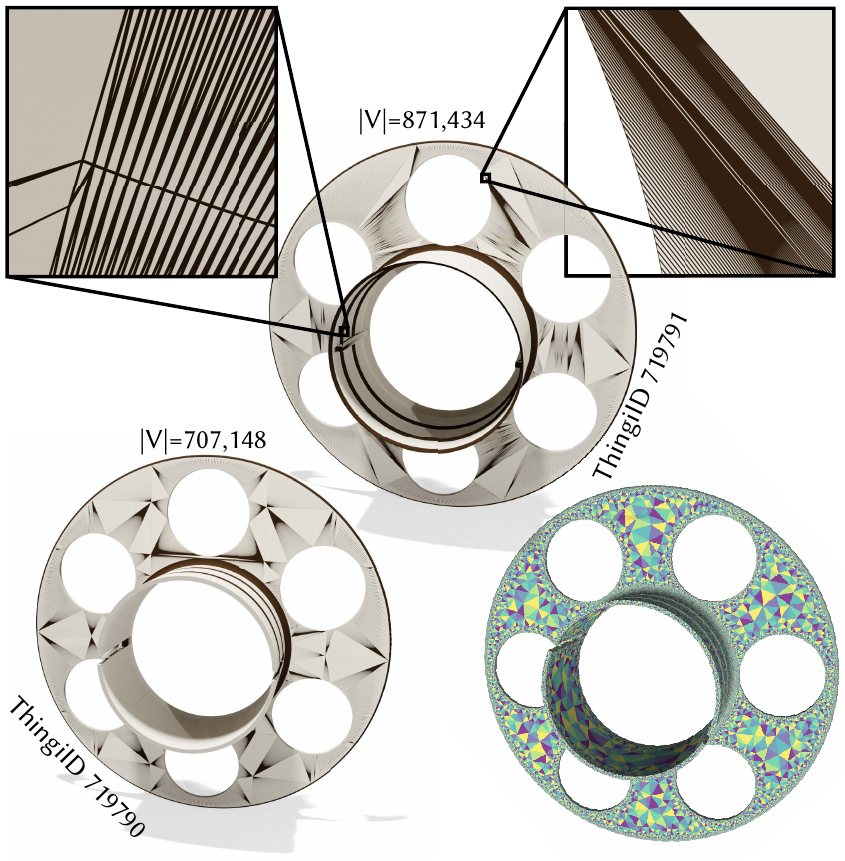}
  \caption{
    We fail to compute an explicit mesh of the common subdivision following Delaunay refinement on one Thingi10k model (\figloc{top}). Its common subdivision would contain 34 million vertices and our program runs out of memory.
We succeed on a nearly identical model (\figloc{bottom}), whose common subdivision contains merely 27 million vertices.
    \label{fig:BadMesh}
  }
\end{figure}

\section{Limitations and Future Work}
\label{sec:FutureWork}

\paragraph{Limitations}
The common subdivisions that we compute after Delaunay refinement can be quite large: the mean increase in \(|V|\) is \(20x\), and the \(95^{\text{th}}\) percentile increase is \(45x\). We emphasize again that the common subdivision is not generally a high-quality mesh anyway: one should perform numerical computations on the intrinsic triangulation instead. The intrinsic mesh has much higher element quality and is generally much smaller with a mean increase in \(|V|\) of \(3.7x\) and \(95^{\text{th}}\) percentile increase of \(7.8x\). However, for applications that rely on the common subdivision (\eg{ \secref{AttributeTransfer}}), it would still be beneficial to explore strategies for simplifying \(S\).

A related issue is that Delaunay refinement sometimes generates meshes with many small triangles. One can prove that Delaunay refinement in the plane produces \emph{well-graded} meshes, meaning essentially that it only places small triangles in regions with small features, and our Delaunay refinement on surfaces seems to behave similarly. Nonetheless, on poorly-conditioned input meshes, Delaunay refinement can insert many small triangles. This can cause problems for diffusion-based algorithms (\eg{} the logarithmic map computation in \figref{LogMapVectorField}), which use the mean edge length to determine a suitable diffusion time. We found that computing the diffusion time on the original mesh and then performing diffusion on the intrinsic triangulation produced the best results.

\paragraph{Hybrid Data Structures}
At this point, there are several intrinsic triangulation data structures, but no single one is perfect:
\begin{itemize}
   \item Overlay (explicit) provides exact connectivity; flipping can be slow; no vertex insertion.
   \item Signposts (implicit) provide inexact connectivity; flipping and vertex insertion are both fast.
   \item Integer coordinates (implicit) provide exact connectivity; flipping is fast; vertex insertion can be slow.
\end{itemize}

We propose a good way to get the best of all worlds would be to use a hybrid signpost + integer coordinate data structure.  This is fully implicit, so you don't pay the \(O(n^2)\) cost when you have \(O(n)\) edges crossing \(O(n)\) edges.  But, flipping and insertion are both fast, and connectivity is exact, if you accept the inserted locations.

Even further in the implicit direction, storing edge lengths is an ``optimization'' in our data structure.
One could \emph{just} store the normal coordinates and original triangulation, recovering edge length whenever necessary via a layout operation.
This is appealing, since it is truly an integer-only representation for intrinsic triangulations.

\paragraph{General geodesic curves}
It would also be natural to use this machinery as a representation for general geodesic curves on surfaces, which commonly arise in geometry processing tasks such as cutting, segmentation, \etc{}.

\section{Acknowledgments}
This work was supported by a Packard Fellowship, NSF
Award 1717320, DFG TRR 109, an NSF Graduate Research Fellowship, and gifts from Autodesk, Adobe, and Facebook.

\bibliographystyle{ACM-Reference-Format}
\bibliography{NormalCoordinates}

\appendix

\section{Pseudocode}
\label{app:Pseudocode}

We assume all algorithms have access to triangulations \(\T0\) and \(\T1\), their edge lengths \(\Tl0, \Tl1\), the normal coordinates \(n\), and the roundabouts \(r\).
\newpage
\begin{figure}
  \includegraphics{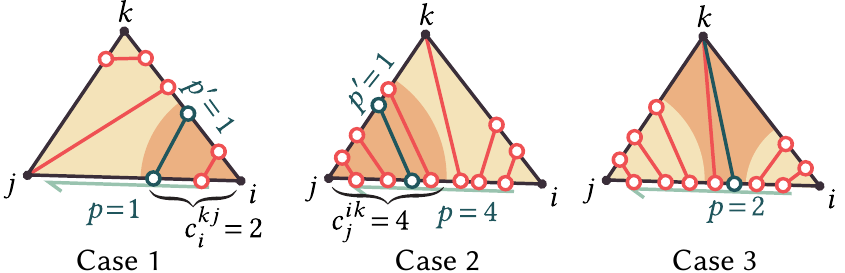}
  \caption{A curve entering triangle \(jik\) along edge \(ij\) can proceed in 3 ways: it can exit along edge \(ik\), in which case it is counted by \(c_i^{kj}\) (\figloc{left}); it can exit along edge \(kj\), in which case it is counted by \(c_j^{ik}\) (\figloc{center}); or it can terminate at vertex \(k\)(\figloc{right}). This forms the core of algorithm \Proc{TraceFrom} [Reproduced from \figref{TraceFromCases} for convenience]. \label{fig:TraceFromCasesAppendix}}
\end{figure}
\begin{algo}{\Proc{TraceFrom}$(\zeta)$}
  \label{alg:TraceFrom}
  \begin{algorithmic}[1]
    \InputConditions{Any combinatorial crossing $\zeta = (\halfedge{ij}, p)$ along some curve \(\gamma\) lying along \(\T1\).}
    \OutputConditions{The half of the curve $\gamma$ as a sequence of points $(\zeta_0, \zeta_1,\dots, \zeta_n, k)$ along $M$, where \(\zeta = \zeta_0\) and \(k\) is the vertex at which \(\gamma\) terminates.}
    \State {\(\mathsf{currentHalfedge} \gets \halfedge{ij}\)}
    \State {\(\gamma \gets [(\mathsf{currentHalfedge}, p)]\)}
    \While {\textsf{True}} \Comment{Walk until the curve terminates at a vertex}
    \State{\LeftComment{Let \(i\) and \(j\) refer the the tail and tip of the current halfedge}}
    \State {\(\halfedge{ij} \gets \mathsf{currentHalfedge}\)}
    \State {\(k \gets \Proc{OppositeVertex}(\Proc{Twin}(\mathsf{currentHalfedge}))\)}
    \If {\(p < c_i^{kj}\)}
    \Comment{Case 1 of \figref{TraceFromCasesAppendix} (\(\gamma\) goes right)}
        \State {\(\mathsf{currentHalfedge} \gets \halfedge{ik}\)} \Comment{Move to \(\halfedge{ik}\)}
        \State {\(p \gets p\)}
        \State {\(\Proc{Append}(\gamma, (\mathsf{currentHalfedge}, p))\)}
    \ElsIf {\(p \geq n_{ij} - c_j^{ik}\)} \Comment{Case 2 of \figref{TraceFromCasesAppendix} (\(\gamma\) goes left)}
      \State {\(\mathsf{currentHalfedge} \gets \halfedge{kj}\)} \Comment{Move to \(\halfedge{kj}\)}
      \State {\(p \gets n_{kj} + p - n_{ij}\)}
      \State {\(\Proc{Append}(\gamma, (\mathsf{currentHalfedge}, p))\)}
    \Else
      \Comment{Case 3 of \figref{TraceFromCasesAppendix} (\(\gamma\) ends at \(k\))}
      \State{\Return \((\gamma, k)\)}
    \EndIf
    \EndWhile
  \end{algorithmic}
\end{algo}

\begin{figure}
  \includegraphics{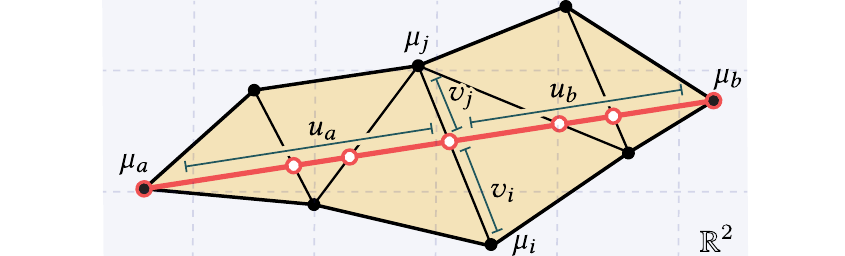}
  \caption{In procedure \Proc{ExtractCurve}, we compute barycentric coordinates by laying out a triangle strip in the plane [Reproduced from \figref{TriangleStripLayout} for convenience]. \label{fig:TriangleStripLayoutAppendix}}
\end{figure}
\begin{algo}{\Proc{ExtractCurve}$(\zeta)$}
  \label{alg:ExtractCurve}
  \begin{algorithmic}[1]
    \InputConditions{Any combinatorial crossing $\zeta = (\halfedge{ij}, p)$ along some curve \(\gamma\)}
    \OutputConditions{The entire trajectory of \(\gamma\) as a sequence of geometric crossings $(a, z_0, z_1,\dots, z_k, b)$ along $M$.}
    \If {runs along edge}
    \State {\Return $e$}
    \Else
    \State {\(\gamma_{\textsf{front}}, b \gets \Proc{TraceFrom}(\zeta)\)} \Comment{Trace forwards along \(\gamma\)}
    \State {\(\gamma_{\textsf{back}}, a \gets \Proc{Reverse}(\Proc{TraceFrom}(\overline{\zeta}))\)} \Comment{Trace backwards}
    \State {\(\gamma_{\textsf{Combinatorial}} \gets \Proc{Append}(\gamma_{\textsf{back}}, \gamma_{\textsf{front}})\)}
    \State{\LeftComment{Compute positions in \(\mathbb{R}^2\) for the triangle strip containing \(\gamma\)}}
    \State {\(\mu \gets \Proc{LayOutTriangleStrip}(\gamma_{\textsf{Combinatorial}})\)}
    \State {\(\gamma_{\textsf{Geometric}} \gets []\)}
    \For {\(\zeta = (\halfedge{ij}, p) \in \gamma_{\textsf{Combinatorial}}\)}
    \State {\LeftComment{Find the intersection of \(ab\) and \(ij\) in the plane (\figref{TriangleStripLayoutAppendix})}}
    \State {\(u, v \gets \Proc{IntersectionBarycentric}(\mu_a, \mu_b, \mu_i, \mu_j)\)}
    \State {\(z \gets (\halfedge{ij}, p, u, v)\)}
    \State {\Proc{Append}(\(\gamma_{\textsf{Geometric}}, z\))}
    \EndFor
    \State {\Return $(a, \gamma_{\textsf{Geometric}}, b), \gamma_{\textsf{Combinatorial}}$}
    \EndIf
  \end{algorithmic}
\end{algo}

\begin{algo}{\Proc{ExtractEdge}$(ab)$}
  \label{alg:ExtractEdge}
  \begin{algorithmic}[1]
    \InputConditions{An edge \(ab \in \E0\)}
    \OutputConditions{The entire trajectory of \(ab\) as a sequence of geometric crossings $(a, z_0, z_1,\dots, z_l, b)$ along $M$}
    \State{\(k \gets \) local index of \(\halfedge{ab}\) about vertex \(a\)}
    \Comment{Find preceding halfedge }
    \State{\(\halfedge{ai} \gets \text{argmax}_{\he{ai}} \{r_{\he{ai}}\;:\; r_{\he{ai}} \leq k\}\)} \Comment{Might wrap cyclically}
    \If {\(r_{\he{ai}} = k\) and \(n_{ai} = -1\)} \Comment{Shared edge}
    \State{\Return \(\halfedge{ai}\)}
    \Else
    \State{\(p \gets k - r_{\he{ai}}-1\)}
    \State{\(\gamma \gets \Proc{ExtractCurve}(\Proc{Next}(\halfedge{ai}), p)\)}
    \State{\LeftComment{If \(\gamma\) does not end at a vertex of \(\V0\) we must keep tracing}}
    \While {\(\gamma\) does not end at a vertex of \(\V0\)}
    \State{\(i \gets \) endpoint of \(\gamma\)}
    \State{\LeftComment{Our curves only pass through vertices inserted via edge splits. Hence, there is a unique other crossing \(\zeta\) emanating from \(j\) that we must trace along}}
    \State{\(\zeta \gets \) other crossing emanating from \(j\)}
    \State{\((i, z_1, \ldots, z_s, j) \gets \Proc{ExtractCurve}(\zeta)\)}
    \State{\LeftComment{\Proc{ExtractCurve} returns geometric crossings whose points have barycentric coordinates \(u\) computed relative to endpoints \(i\) and \(j\). We should return barycentric coordinates relative to \(a\) and \(b\) instead. Since \(\Tpos0i\), \(\Tpos0j\) give barycentric coordinates for \(i\) and \(j\) along \(ab\), this amounts to linear interpolation of those coordinates}}
    \State{\(\Proc{AdjustBarycentricCoordinates}(z_1, \ldots, z_s)\)}
    \State{\(\Proc{Append}(\gamma, (z_1, \ldots, z_s, j))\)}
    \EndWhile
    \State{\Return \(\gamma\)}
    \EndIf
  \end{algorithmic}
\end{algo}

\begin{figure}
  \includegraphics{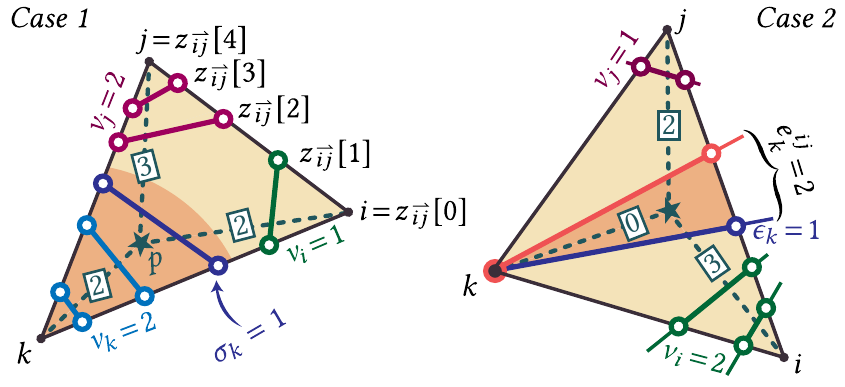}
  \caption{In \Proc{SplitFace}, we do a sequence of line-side tests to compute a value of \(\nu\) at each corner.\label{fig:FaceSplitAlgo}}
\end{figure}

\begin{algo}{\Proc{SplitFace\_Case1}$(ijk, u)$}
  \label{alg:SplitFace1}
  \begin{algorithmic}[1]
    \InputConditions{The location to insert a vertex on $\T1$, as barycentric coordinates $u$ in a face $ijk \in \F1$.}
    \OutputConditions{An updated integer coordinate intrinsic triangulation.}
    \State {\LeftComment{Gather all geometric crossings}}
    \For {\(\zeta = (\halfedge{ij}, p) \in \Proc{CombinatorialCrossings}(ijk)\)}
    \State {\(z_{\he{ij}}[p+1] \gets \Proc{ExtractGeometricCrossing}(\zeta)\)}
    \EndFor
    \State \LeftComment{Compute new normal coordinates}
    \For{\(\corner{i}{jk} \in \Proc{CornersOf}(ijk)\)} \Comment{for each corner}
    \State{\LeftComment{Identify which corner curves, if any, contain u}}
      \State {\(\mathcal{C}_i \gets \left\{\xi:  0 \leq \xi < c_i^{jk}, u \in \Proc{Triangle}\left(i, z_{\he{ij}}[\xi], z_{\he{ik}}[\xi]\right)\right\}\)}
      \State {$\nu_{i} \gets \min \big( c_i^{jk} \cup \mathcal{C}_i \big)$} \Comment{Take the closest such corner}
      \State {\(\sigma_{i} \gets c_i^{jk} - \nu_i\)} \Comment{``Slack'' left at corner}
    \EndFor

    \State \LeftComment{In exact arithmetic, only one \(\sigma_i\) may be nonzero. In floating point multiple could be nonzero, so we keep the biggest and round the others to zero}
    \If {\(\sigma_i \geq \sigma_j, \sigma_k\)}
        \State {\(\sigma_j, \sigma_k \gets 0, 0\)}
    \ElsIf {\(\sigma_j \geq \sigma_k, \sigma_i\)}
    \State {\(\sigma_k, \sigma_i \gets 0, 0\)}
    \ElsIf {\(\sigma_k \geq \sigma_i, \sigma_j\)}
    \State {\(\sigma_i, \sigma_j \gets 0, 0\)}
    \EndIf

    \For{\(\corner{i}{jk} \in \Proc{CornersOf}(ijk)\)} \Comment{include slack crossings}
      \State {$n_{pi} \gets \nu_{i} + \sigma_j + \sigma_k $} 
    \EndFor
      
    \State \LeftComment{Compute everything else}
    \State {$r \gets \Proc{UpdateRoundabouts}(n_{pi}, n_{pj}, n_{pk}, r)$} \Comment{\eqref{RoundaboutUpdate}}
    \State {$\Tl1_{pi}, \Tl1_{pj}, \Tl1_{pk} \gets \Proc{UpdateEdgeLengths}(\Tl1, u)$} \Comment{\eqref{NewEdgeLengths}} 
    \State {$R \gets \Proc{RegionFromNormalCoordinates}(n_{pi}, n_{pj}, n_{pk})$}
    \State {$\Tpos0p \gets \Proc{RecoverBarycentric}(R, Z, u)$}\Comment{\appref{RecoverBarycentric}} 
  \end{algorithmic}
\end{algo}

\begin{algo}{\Proc{SplitFace\_Case2}$(ijk, u)$}
  \label{alg:SplitFace2}
  \begin{algorithmic}[1]
    \InputConditions{The location to insert a vertex on $\T1$, as barycentric coordinates $u$ in a face $ijk \in \F1$. In Case 1, \(ijk\) must be oriented so that \(n_{ij} \geq n_{jk} + n_{ki}\)}
    \OutputConditions{An updated integer coordinate intrinsic triangulation.}
    \State {\LeftComment{Gather all geometric crossings}}
    \For {\(\zeta = (\halfedge{ij}, p) \in \Proc{CombinatorialCrossings}(ijk)\)}
    \State {\(z_{\he{ij}}[p+1] \gets \Proc{ExtractGeometricCrossing}(\zeta)\)}
    \EndFor
    \State \LeftComment{Compute new normal coordinates}
    \For{\(\corner{i}{jk} \in \Proc{CornersOf}(ijk)\)} \Comment{for each corner}
    \State{\LeftComment{Identify which corner curves, if any, contain u}}
      \State {$\mathcal{C}_i \gets \left\{\xi:  0 \leq \xi < c_i^{jk}, u \in \Proc{Triangle}\left(i, z_{\he{ij}}[\xi], z_{\he{ik}}[\xi]\right)\right\}\)}
      \State {$\nu_{i} \gets \min \big( c_i^{jk} \cup \mathcal{C}_i \big)$} \Comment{Take the closest such corner}
      \State {\(\sigma_{i} \gets c_i^{jk} - \nu_i\)} \Comment{``Slack'' left at corner}
      \EndFor
      \State{\LeftComment{Note that \(\nu_k = \sigma_k = 0\)}}

      \If {\(\sigma_i \geq \sigma_j\)} \Comment{Ensure that at most one \(\sigma\) is nonzero}
        \State{\(\sigma_j \gets 0\)} 
      \ElsIf {\(\sigma_j \geq \sigma_i\)}
        \State{\(\sigma_i \gets 0\)} 
      \EndIf

      \State{\LeftComment{Check for intersections with emanating edges}}
      \If {\(\nu_{i} < c_i^{jk}\)} \Comment{In corner \(i\)}
        \State \(\nu_{j} \gets \nu_{j} + e_k^{ij}\)
      \ElsIf {\(\nu_{j} < c_j^{ki}\)} \Comment{In corner \(j\)}
        \State \(\nu_{i} \gets \nu_{j} + e_k^{ij}\)
      \Else \Comment{In middle of fan region}
      \State{\LeftComment{Identify which emanating curves, if any, contain u}}
      \State {$\mathcal{E}_k \gets \left\{\xi:  0 \leq \xi < e_k^{ij}, u \in \Proc{Triangle}\left(i, z_{\he{ij}}[\xi +c_i^{jk}], k\right)\right\}\)}
      \State{\(\epsilon_k \gets \min\left( e_k^{ij} \cup \mathcal{E}_k \right)\)}  \Comment{Take the curve closest to \(i\)}
      \State {\(\nu_{i} \gets \nu_{i} + \epsilon_k\)} \Comment{Edge \(pi\) crosses the first \(\epsilon_k\) such curves}
      \State {\(\nu_{j} \gets \nu_{j} + e_i^{jk} - \epsilon_k\)} \Comment{Edge \(pj\) crosses the rest}
      \EndIf
    \For{\(\corner{i}{jk} \in \Proc{CornersOf}(ijk)\)} \Comment{include opposite corners}
      \State {$n_{pi} \gets \nu_{i} + \sigma_j + \sigma_k $} 
    \EndFor
      
    \State \LeftComment{Compute everything else}
    \State {$r \gets \Proc{UpdateRoundabouts}(n_{pi}, n_{pj}, n_{pk}, r)$} \Comment{\eqref{RoundaboutUpdate}}
    \State {$\Tl1_{pi}, \Tl1_{pj}, \Tl1_{pk} \gets \Proc{UpdateEdgeLengths}(\Tl1, u)$} \Comment{\eqref{NewEdgeLengths}} 
    \State {$R \gets \Proc{RegionFromNormalCoordinates}(n_{pi}, n_{pj}, n_{pk})$}
    \State {$\Tpos0p \gets \Proc{RecoverBarycentric}(R, Z, u)$}\Comment{\appref{RecoverBarycentric}} 
  \end{algorithmic}
\end{algo}

\begin{algo}{\Proc{SplitEdge}$(\halfedge{ij}, u)$}
  \label{alg:SplitEdge}
  \begin{algorithmic}[1]
    \InputConditions{The location to insert a vertex on \(\T1\), as barycentric coordinates \(u\) on a halfedge \(\halfedge{ij} \in \H1\)}
    \OutputConditions{An updated integer coordinate intrinsic triangulation}
    \If {\(n_{ij} \geq 0\)}
    \State{\(k \gets \Proc{OppositeVertex}(\halfedge{ij})\)}
    \State{\(\Proc{SplitFace}(ijk, (u_i, u_j, 0))\)}
    \State{\(\Proc{FlipEdge}(ij)\)}
    \Else
    \State{\(k \gets \Proc{OppositeVertex}(\halfedge{ij})\)}
    \If {\(\Proc{InInterior}(ij)\)}
    \State{\(l \gets \Proc{OppositeVertex}(\halfedge{ji})\)}
    \EndIf
    \State{\(\T1 \gets \text{insert a vertex} \; p \; \text{along} \; ij\)} \Comment{Update combinatorics}
    \State{\(n_{pj}, n_{pi} \gets n_{ij}, n_{ij}\)} \Comment{Compute new normal coordinates}
    \State{\(n_{pk} \gets \max(n_{ki}, n_{jk}, 0)\)}

    \State \LeftComment{Compute everything else}
    \State {$r \gets \Proc{UpdateRoundabouts}(n_{pi}, n_{pj}, n_{pk}, n_{pl}, r)$} \Comment{\eqref{RoundaboutUpdate}}
    \State {$\Tl1_{pi}, \Tl1_{pj}, \Tl1_{pk}, \Tl1_{pl} \gets \Proc{UpdateEdgeLengths}(\Tl1, u)$} \Comment{\eqref{NewEdgeLengths}} 
    \State{\(\Tpos0p \gets u_i \Tpos0i + u_j \Tpos0j\)}
    \EndIf
  \end{algorithmic}
\end{algo}

\begin{algo}{\Proc{RemoveVertex}$(i)$}
  \label{alg:RemoveVertex}
  \begin{algorithmic}[1]
    \InputConditions{An inserted vertex $i$}
    \OutputConditions{Updated triangulation $i$ removed}
    \While {$i$ has degree $> 3$}
    \State {\(ij \gets \) flippable edge incident on $i$}
    \State {$\Proc{FlipEdge}(ij)$}
    \EndWhile
    \State {$\Proc{DeleteVertexAndIncidentEdges}(i)$}
  \end{algorithmic}
\end{algo}

\newcommand{\Vidx}{\mathcal{I}^v}

\begin{algo}{\Proc{ComputeCommonSubdivision}$()$}
  \label{alg:ComputeCommonSubdivision}
  \begin{algorithmic}[1]
    \InputConditions{Nothing beyond the usual data (\ie{} \(\T0, \T1, \ldots\))}
    \State{\LeftComment{Index common subdivision vertices}}
    \State{\(\Vidx \gets \text{index common subdivision vertices}\)}
    \State{\LeftComment{Compute connectivity}}
    \State{\(\textsf{polygons} \gets []\)}
    \For{\(ijk \in \F1\)} \Comment{Always orient such that \(n_{ij} \geq n_{jk}, n_{ki}\)}
    \If {\(e_k^{ij} = 0\)}
    \State{\(\Proc{Append}(\textsf{polygons}, \Proc{SubdivideFace\_Case1}(ijk, \Vidx))\)}
    \Else
    \State{\(\Proc{Append}(\textsf{polygons}, \Proc{SubdivideFace\_Case2}(ijk, \Vidx))\)}
    \EndIf
    \EndFor
    
    \State{\LeftComment{Compute intersection geometry. We denote the locations on \(\T0\) by \(Q^0\) and the locations on \(\T1\) by \(Q^1\)}}
    \For{\(i \in \V1\)} \Comment{Vertices of \(\T1\)}
      \State{\(k \gets \Vidx_i\)} \Comment{Index of vertex in \(S\)}
        \State{\(Q_k^0 \gets q_i^0\)} \Comment{Location on \(\T0\) computed when \(i\) was inserted}
        \State{\(Q_k^1 \gets (i, 1)\)}  \Comment{Location on \(\T1\) is just \(i\) itself}
    \EndFor
    \For{\(ab \in \E0\)} \Comment{Edge intersections}
    \State{\(\gamma = (a, z_1, \ldots, z_l, b) \gets \Proc{ExtractEdge}(ab)\)}
    \For{\(z = (\halfedge{ij}, p, u, v) \in \gamma\)}
    \State{\(k \gets \Vidx_{\he{ij}}[p+1]\)} \Comment{Index of crossing in \(S\)}
    \State{\(Q^0_k \gets (\halfedge{ab}, u)\)} \Comment{Position on \(\T0\) along \(\halfedge{ab}\)}
    \State{\(Q^1_k \gets (\halfedge{ij}, v)\)} \Comment{Position on \(\T1\) along \(\halfedge{ij}\)}
    \EndFor
    \EndFor

    \State{\Return \(\textsf{polygons}, Q^0, Q^1\)}
  \end{algorithmic}
\end{algo}

\begin{algo}{\Proc{DelaunayRefinement}$(\theta_{\min{}})$}
  \label{alg:DelaunayRefinement}
  \begin{algorithmic}[1]
    \InputConditions{A minimum allowed angle \(\theta_{\min{}}\).}
    \OutputConditions{An intrinsic triangulation \(\T1\) whose corner angles are all at least \(\theta_{min}\)}
    \State $\Proc{FlipToDelaunay}()$
    \While {\(\T1\) has triangles with angles less than \(\theta_{\min}\)}
    \State {\(ijk \gets \) any triangle with an angle less than \(\theta_{\min}\)}
    \State{\LeftComment{Find the circumcenter via the exponential map}}
    \State {\(v_c \gets \) circumcenter barycentric coordinates} \Comment{Equations \ref{eq:HomogeneousCircumcenterBary}, \ref{eq:CircumcenterBary}}
    \State{\LeftComment{Barycentric coordinate offset from barycenter to circumcenter}}
    \State {\(\delta v_c \gets v_c - (1/3, 1/3, 1/3)\)}
    \State{\LeftComment{Transform offset to face tangent space}}
    \State {\(V \gets \Proc{BarycentricOffsetToTangentVector}(\delta v_c)\)}
    \State{\LeftComment{Evaluate exponential map from face barycenter}}
    \State \(c \gets \Proc{Exp}(\Proc{Barycenter}(ijk), V)\)
    \If {\(c\) lies inside the mesh}
    \State \(\Proc{InsertCircumcenter}(ijk)\)
    \Else
    \State {\(lm \gets \) boundary edge separating \(c\) from \(ijk\)}
    \State \(p \gets \Proc{SplitEdge}(lm, 0.5)\)
    \State{\LeftComment{Must flip to Delaunay before computing Dijkstra ball}}
    \State $\Proc{FlipToDelaunay}()$
    \State {\LeftComment{Remove inserted vertices from \(lm\)'s diametral ball}}
    \State {\(\mathsf{ball} = \{i \in \V1 : \Proc{DijkstraDistance}(\E1, i, p) < \ell_{lm}\}\)}
    \For {\(i \in \mathsf{ball}\)}
    \State \(\Proc{RemoveVertex}(i)\)
    \EndFor
    \EndIf
    \State $\Proc{FlipToDelaunay}()$
    \EndWhile
  \end{algorithmic}
\end{algo}

\section{Barycentric Coordinates Recovery}
\label{app:RecoverBarycentric}

In \secref{FaceSplit}, we recover the barycentric coordinates of a newly inserted vertex on $\T0$ via interpolation along a polygonal subregion $R$ of triangle $abc \in \F0$. Here we give a full expression for the necessary small linear system.

Precisely, let \(3 \leq \rho \leq 6\) denote the number of corners of \(R\).
Let the \(\smash{m^{\text{th}}}\) corner of \(R\) have barycentric coordinates \(\smash{u^{(m)}_{a}, u^{(m)}_b, u^{(m)}_c}\) on \(abc \in \F0\) and barycentric coordinates \(\smash{v^{(m)}_i, v^{(m)}_j, v^{(m)}_k}\) on \(ijk \in \F1\), all of which are know. We also know the barycentric coordinates \(v_i\) for $p$ in \(ijk\).  We then want to solve for the corresponding \(u_a\) on \(abc\). We proceed in two steps: first, we express \(v\) as a linear combination $\xi$ of the \(v^{(m)}\). Then, we apply this same linear combination to the \(u^{(m)}\) to obtain \(u\). Concretely, we first solve for the minimum-norm solution of the underdetermined system
\begin{equation}
    \begin{pmatrix}
      v^{(0)}_i & v^{(1)}_i & \cdots & v_i^{(\rho)}\\
      v^{(0)}_j & v^{(1)}_j & \cdots & v_j^{(\rho)}\\
      v^{(0)}_k & v^{(1)}_k & \cdots & v_k^{(\rho)}
    \end{pmatrix}
    \begin{pmatrix}
      \xi_0\\\xi_1\\\vdots\\\xi_{\rho}
    \end{pmatrix}
    = \begin{pmatrix}v_i\\v_j\\v_k\end{pmatrix},
\end{equation}
and then set
\begin{equation}
  \label{eq:RecoverInterp}
  u_a := \sum_m u^{(m)}_a \xi_m,\quad
  u_b := \sum_m u^{(m)}_b \xi_m,\quad
  u_c := \sum_m u^{(m)}_c \xi_m.
\end{equation}
Note that while one often seeks a nonnegative $\xi$, any solution will suffice here: we only use \(\xi\) to interpolate in \eqref{RecoverInterp}.

\section{Delaunay Refinement Details}
\subsection{Removing Extra Vertices}
\label{app:RemovingExtraVertices}
When Chew's second algorithm splits an edge, it removes all inserted circumcenters within a geodesic ball centered at the edge's midpoint. These vertices must be removed, but it is okay to removes additional interior inserted vertices. \citet[Section 3.4.2]{Shewchuk:1997:DRMG} observes that the algorithm can only perform finitely many edge splits. As long as one removes all interior inserted vertices within the geodesic ball---and never removes vertices along the boundary---the algorithm will still perform only finitely many edge splits. Hence, it must terminate as usual following the final edge split, even if one removes extra circumcenters during edge splits.

\subsection{Proof of Correctness on Watertight Meshes}
\label{app:DelaunayRefinementProof}

Here we seek to prove that \Proc{DelaunayRefinement} (\algref{DelaunayRefinement}) succeeds, in the basic case of a closed surface with bounded cone angles.
We will not prove the more general boundary case here, but experimentally we observe success on a large dataset (\secref{Robustness}).

\begin{theorem}[Delaunay refinement, no boundary]
  On meshes without boundary, with vertex angle sums at least \(60^\circ\), \algref{DelaunayRefinement} produces a Delaunay mesh with triangle corner angles at least \(30^\circ\).
\end{theorem}

\textsc{Proof.}
  By definition, \Proc{DelaunayRefinement} only terminates when the triangulation is a Delaunay triangulation which satisfies the angle bound, so we just need to prove that termination occurs after a finite number of iterations.
  We will show this by establishing that \Proc{DelaunayRefinement} maintains a minimum spacing between all vertices in the mesh, so the number of insertions is bounded by surface area.
  Our argument will generally follow the planar proof of \citet[Section 3.2.1]{Shewchuk:1997:DRMG}, though extra care is needed in the intrinsic case, where self edges may connect a vertex to itself.

  In particular, we consider the length of the shortest edge in the initial mesh's intrinsic Delaunay triangulation, $\delta := \min_{ij} \ell_{ij}$.
  We will show that the minimum edge length in each subsequent Delaunay triangulations is at least $\delta$.
  Then all vertices must be separated by a distance at least \(\delta\), since \lemref{DelaunayNN}, each vertex is connected to its geodesic nearest neighbor.
  Hence, each vertex is contained in an open disk of radius \(\tfrac 12 \delta\) which is disjoint from all other disks.
  As the input mesh has finite surface area, we conclude that \algref{DelaunayRefinement} can only insert finitely many vertices, and thus must terminate.

  It remains to show that \Proc{DelaunayRefinement} never creates an edge of length less than $\delta$.
  It is convenient to convert the angle bound $\alpha$ to a \emph{circumradius-to-shortest-edge ratio} bound $B = \tfrac{1}{2\sin \alpha}$ \cite[Section 3.1]{Shewchuk:1997:DRMG}.
  Having corner angles at least \(\alpha=30^\circ\), is equivalent to a circumradius-to-shortest-edge ratio of at most \(B = 1\), and  thus we insert the circumcenters of triangles with $B > 1$.

  We proceed by induction.
  All initial edges have length at least $\delta$ by definition.
  Now consider inserting vertex $i$ at the circumcenter of triangle \(jkl\) with circumradius $R$.
  Since we only split triangles with $B > 1$, and \(jkl\)'s edges have length at least \(\delta\), we must have $R > \delta$.
  By \lemref{RefinementIncidentEdges} all new edges in the Delaunay triangulation must be incident on $i$, and since \(jkl\) had an empty geodesic circumcircle, there can be no other vertices within distance $R > \delta$.
  Thus new all edges to other vertices have length at least $\delta$.
  We must now consider self edges connecting the new vertex $i$ to itself.

  Gluing together the two ends of a self edge yields a loop; we will split into cases based on the homotopy class of this loop on the punctured surface (before the insertion of $i$).
  First, note that the loop cannot be contractible to a point, since the original edge is geodesic.
  Then we will split in to two cases: either the loop contracts around a single vertex, or it does not.

  \setlength{\columnsep}{.5em}
  \setlength{\intextsep}{.5em}
  \begin{wrapfigure}[14]{r}{54pt}
    \vspace{-0\baselineskip}\includegraphics{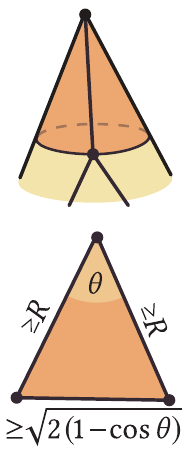}
  \end{wrapfigure}
  If the loop contracts around a single vertex, then the self edge encloses a degree-1 vertex.
  The degree-1 vertex must have distance at least \(R\) to the inserted vertex, and has angle sum at least \(60^\circ\). Thus, by the law of cosines, the length of the self edge must be at least \[\sqrt{R^2 + R^2 - 2R^2 \cos \theta} = R\sqrt{2(1-\cos \theta)}.\]
  Since \(\cos 60^\circ = \tfrac 12\), and \(1-\cos \theta\) is increasing with \(\theta\), this shows that the self edge has length at least \(R\) whenever \(\theta\) is at least \(60^\circ\).
  
  If the loop is not in a homotopy class contractible about a single vertex, then the shortest loop $\gamma_{\min{}}$ in the homotopy class is non-constant.
  By \lemref{LoopTouchesVertex}, we can take $\gamma_{\min{}}$ to touch some vertex $a$, and note that since $\gamma_{\min{}}$ is the \emph{shortest} loop our original self edge must be at least as long as $\gamma_{\min{}}$.
  Then by \lemref{DelaunayNN} $a$ has an edge at least a long as $\gamma_{\min{}}$, and thus the self edge has length $\geq |\gamma_{\min{}}| \geq \delta$.

  Thus, we conclude that \algref{DelaunayRefinement} never introduces an edge of length less than \(\delta\), which means that it must terminate after inserting finitely many vertices.
  \hfill\qed

\begin{lemma}
  For any pair of vertices \(i,j \in V\), let \(\Gamma_{ij}\) be the set of non-constant geodesics connecting \(i\) to \(j\). Then \[d_{ij} := \inf_{\gamma \in \Gamma_{ij}} \text{length}(\gamma) > 0.\]
  \label{lem:FiniteLength}
\end{lemma}
\begin{proof}
  This follows directly from \cite[Proposition 1]{Indermitte:2001:VDP}, which states that for any \(L>0\), the number of geodesic arcs from \(i\) to \(j\) of length at most \(L\) is finite. Since any geodesic of length 0 is constant, and thus not in \(\Gamma_{ij}\), this implies that \(d_{ij} > 0\).
\end{proof}

\begin{lemma}
  For any vertex \(i \in V\), the intrinsic Delaunay triangulation contains an edge to \(i\)'s nearest neighbor.
  \label{lem:DelaunayNN}
\end{lemma}
\begin{proof}
This is a standard result, which we include for completeness.
Let \(j\) be \(i\)'s nearest neighbor, \ie{} \(j := \text{argmin}_j d_{ij}\).
Note that \(j\) may equal \(i\), and \(d_{ij} > 0\) by \lemref{FiniteLength}.
Consider the disk \(D\) of radius \(d_{ij}\) centered at \(i\).
Since \(j\) is \(i\)'s nearest neighbor, \(D\) contains no vertex other than \(i\).
Thus, the circle which goes through \(i\) and \(j\) and is tangent to \(D\) at \(j\) has empty interior, and its boundary contains no vertices other than \(i\) and \(j\).
We conclude that \(ij\) is in the Delaunay triangulation \cite[Definition 3]{Bobenko:2007:ADL}.
\end{proof}

\begin{lemma}
  All edges created in \Proc{DelaunayRefinement} following the insertion of a vertex \(i\) and flipping to Delaunay are incident on \(i\).
  \label{lem:RefinementIncidentEdges}
\end{lemma}
\begin{proof}
  Again, we follow the planar proof of \citet[Lemma 12]{Shewchuk:1997:DRMG}. We wish to prove that all Delaunay edges which are not incident on \(i\) were Delaunay before inserting \(i\). This follows from the fact that edges of a Delaunay triangulation satisfy an empty circumcircle condition~\cite[Definition 3]{Bobenko:2007:ADL}. If an edge's circumcircle is empty after inserting vertex \(i\), it must have been empty before too, so the edge was already Delaunay.
\end{proof}

\begin{lemma}
  Any geodesic loop $\gamma$ is isotopic to a geodesic loop $\gamma'$ of the same length which touches a vertex.
  \label{lem:LoopTouchesVertex}
\end{lemma}
\begin{proof}
  $\gamma$ can ``slide'' until it touches a vertex without changing its length.
  Precisely, consider a unit-speed motion of \(\gamma\) within the surface along its outward normal direction.
  During the motion, $\smash{\tfrac{d}{dt}|\gamma|} = \smash{\int_\gamma \kappa(s) \; ds}$, where $\kappa$ is the geodesic curvature of \(\gamma\).
  Since $\gamma$ is a geodesic, $\kappa = 0$: its length does not change.
  Thus we can construct \(\gamma'\) by sliding $\gamma$ along the surface until it touches a vertex.
\end{proof}
As an aside, we note that geodesic loops which do not touch a vertex only occur in non-generic configurations.

\begin{figure}[b]
  \includegraphics{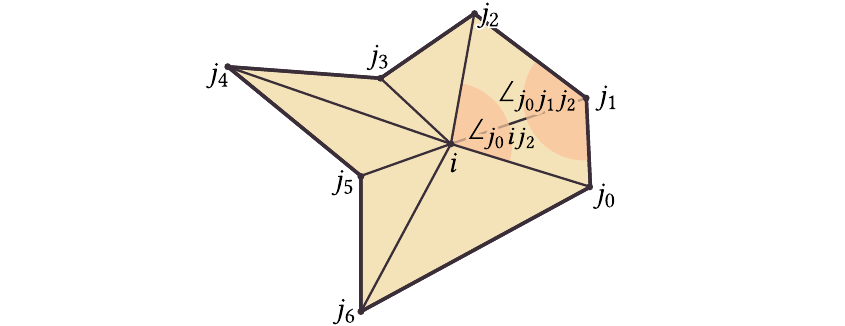}
  \caption{
    The relevant angles for \thmref{VertexRemovalFlipping}.
    \label{fig:VertexRemovalAngles}
  }
\end{figure}

\section{Simplicial Vertex Removal}
\label{app:VertexRemovalProof}

In \secref{VertexRemoval}, we consider removing a vertex by flipping edges until the vertex has degree three and then deleting it.
Past work has also proposed this approach in the purely-topological setting \cite[Section 5.4]{Schaefer:2002:ANC}, but here we must respect geometric constraints.
In particular, edges can only be flipped geometrically if they are contained in a convex quadrilateral (\secref{EdgeFlip}).
Here we prove that flipping edges to remove vertices is indeed a viable strategy in the Euclidean setting as well: one can always find an edge to flip.

\begin{theorem}[Vertex Removal, simplicial]
  If a vertex $i$ in a simplicial complex has cone angle $2 \pi$ and degree $d>3$, then some edge $ij$ incident on $i$ can be flipped to decrease the degree of $i$.
  \label{thm:VertexRemovalFlipping}
\end{theorem}
\begin{proof}
  Recall that an edge can be flipped if both endpoints will have degree at least $1$ after the flip, and the edge is contained in a convex quadrilateral (\secref{EdgeFlip}).
  As always, the convex quadrilateral is defined in the sense of the intrinsic geometry determined by edge lengths.
  The endpoint degree constraint is automatically satisfied on a simplicial complex, so we only need to show that the geometric convexity constraint is satisfied, which is equivalent to showing that all angles of the edge's quadrilateral are at most $\pi$.

  Denote the neighboring vertices of $i$ as $j_k$, with $j_{k+1}$ \etc{} implicitly indexed modulo the vertex degree $d$ (\figref{VertexRemovalAngles}).
  The outer angles $\angletriplet{i}{j_{k-1}}{j_k}$ and $\angletriplet{i}{j_{k+1}}{k}$ are corners of Euclidean triangles, and thus are necessarily at most $\pi$, so we need to find an edge $i j_k$ for which the angles $\angletriplet{j_{k-1}}{i}{j_{k+1}}$ and $\angletriplet{j_{k+1}}{j_k}{j_{k-1}}$ are also at most $\pi$.

  First we consider the inner corners $\angletriplet{j_{k-1}}{i}{j_{k+1}}$. At most two of these angles can be greater than $\pi$.
  To see why, suppose there were three $\angletriplet{j_{k-1}}{i}{j_{k+1}} > \pi$. Since the degree of $i$ is $d>3$, then some pair of those three large angles would correspond to disjoint angular sectors around the vertex, and summing their angles yields a value greater than $2\pi$, which is impossible because the angle sum of $i$ is $2 \pi$. Thus all but at most two of the edges incident on $i$ have inner corners with angle at most $\pi$.

  Likewise, at least three of the outer corners \(\angletriplet{j_{k+1}}{j_k}{j_{k-1}}\) are at most \(\pi\). This is because the sum of all \(d\) outer corners must be $(d-2) \pi$. Since they are nonnegative, at most $d-3$ of them can be strictly greater than $\pi$, implying that at least $3$ will be $\pi$.

  Thus at least three outer corners are at most $\pi$, and at most two of the inner corners are \emph{not} at most $\pi$, so there must be at least one edge for which both the inner and outer corners are at most $\pi$.
  This edge can then be flipped, reducing the vertex degree.
\end{proof}

Importantly, this proof \emph{does not} handle the full general case of a $\Delta$-complex, where there may exist self-edges which cause flips to not make progress.
However, we note that \citet[Appendix A]{Sharp:2020:YCF} proves that a similar flip-removal strategy works in the case of a $\Delta$-complex, and we conjecture that an analogous technique could be applied to generalize \thmref{VertexRemovalFlipping}.
Also, note that the ``equality'' case of \thmref{VertexRemovalFlipping} is a possibility, such as a degree four cross configuration where all angles $ = \pi / 2$. 
Fortunately the resulting skinny triangle after the edge is a non-issue, because the center vertex is about to be removed.

\end{document}